\documentclass[final,leqno]{article}
\usepackage[utf8]{inputenc}
\usepackage[english]{babel}
\usepackage{times}
\usepackage{amsmath}
\usepackage{amsfonts}
\usepackage{amssymb}
\usepackage{amsthm}
\usepackage{graphicx}
\usepackage{graphics}
\usepackage{epsfig}
\usepackage{subfig}
\usepackage[section]{placeins}
\usepackage{bm}
\usepackage{fancyhdr}
\usepackage{url}
\usepackage{caption}
\usepackage{cleveref}
\usepackage{float}
\usepackage[letterpaper,margin=1in]{geometry}
\usepackage[colorinlistoftodos]{todonotes}

\newtheorem{theorem}{Theorem}[section]
\newtheorem{example}[theorem]{Example}

\newcommand{\comment}[1]{}     
\begin{document}
\title{The role of short-term immigration on disease dynamics: An SIR model with age-structure}
\maketitle

\begin{center}
\author{Fabio Sanchez\footnote{Corresponding author: \url{fabio.sanchez@ucr.ac.cr}\\
Present address: Centro de Investigaci\'on en Matem\'atica Pura y Aplicada (CIMPA), Escuela de Matem\'atica, Universidad de Costa Rica. San Pedro de Montes de Oca, San Jos\'e, Costa Rica, 11501.} and
Juan G. Calvo\footnote{Centro de Investigaci\'on en Matem\'atica Pura y Aplicada (CIMPA), Escuela de Matem\'atica, Universidad de Costa Rica. San Pedro de Montes de Oca, San Jos\'e, Costa Rica, 11501. Email: \url{juan.calvo@ucr.ac.cr}}}
\end{center}




\begin{abstract}
\noindent We formulate an age-structured nonlinear partial differential equation model that features {\it short-term immigration} effects in a population. Individuals can immigrate into the population as any of the three stages in the model: {\it susceptible}, {\it infected} or {\it recovered}. Global stability of the {\it immigration-free} and {\it infection-free} equilibria is discussed. A generalized numerical framework is established and specific {\it short-term immigration} scenarios are explored.
\end{abstract}






\section{Introduction} \label{sec:intro}
Propagation of infectious diseases due to immigration can have adverse effects in a population. In the last decade, numerous events around the world \cite{HondCaravan,Nica,Vene,Vene2,eurocrisis,prcrisis} forced individuals to find themselves in a desperate situation where migration to other countries is the only viable alternative for surviving.

An expected consequence of {\it short-term immigration} is being an economic and public health burden in a population. These population movements, sometimes massive \cite{HondCaravan}, can cause financial burden, overcrowding, cultural clashes, among other issues if the country receiving them is not prepared to do so, or if it is incapable of sustaining such influx of individuals. We present an epidemiological approach with age-structure that takes into account {\it short-term immigration} and the effects it may have on a population regarding disease dynamics. 

Traditionally, the central objective in epidemic models is the study of the {\it the basic reproductive number}, $\mathcal{R}_0$, which represents the number of secondary infections caused by a infectious individual when introduced into a mostly susceptible population. 
Typically, when $\mathcal{R}_0>1$, the disease prevails and the endemic equilibrium becomes asymptotically stable. The possibility of infected individuals entering a (susceptible) population can lead to an endemic state ($\mathcal{R}_0>1$) and cause an outbreak in a population. Moreover, having $\mathcal{R}_0<1$ becomes a very difficult task. These situations bring enormous challenges to public health officials when these \lq\lq bursts\rq\rq of individuals enter a population in droves. 
We will introduce a general numerical approach and will explore numerical simulations under different scenarios with special circumstances, based on the aforementioned recent world events and the approach by \cite{franceschetti2008}.

Previous work on age-structured models focuses on studying the long term dynamics of the system, existence and stability of its equilibria and strategies around the {\it basic reproductive number} \cite{allen2008mathematical,holmes1994partial,kribs2002,martcheva2003progression,shim2006,sanchez2018x}. The age-structured model by \cite{franceschetti2008} includes permanent immigration into the population. In contrast, we explore an age-structured model that incorporates {\it short-term immigration} effects that can lead to a population having a large outbreak caused by the introduction of an infectious disease foreign to the local population, or an absent disease where the population has some level of immunity.

The manuscript is organized as follows. In Section \ref{sec:data}, we describe some of the recent events that have led to large influx of migrants. In Section \ref{sec:model}, we describe the partial differential equation model. In Section \ref{sec:freesteadyState}, we study the {\it immigration-free} and {\it infection-free} non-uniform steady state distribution, its stability, and define the {\it basic reproductive number}, $\mathcal{R}_0$, when there is no immigration. Then, a general numerical framework is presented in Section \ref{sec:exp} that explores various immigration scenarios and its effect on a presumably susceptible population. Lastly, in Section \ref{sec:conc} we discuss our results of the model.

\section{Recent world events} \label{sec:data}
Recent world events have caused populations to migrate to other countries looking for safety, jobs, or simply survival \cite{HondCaravan,Nica,Vene,Vene2,eurocrisis,prcrisis}. In 2017, one of the worst natural disasters to hit Puerto Rico, category 5 hurricane Maria, destroyed the country, leaving a death-toll of close to 3000 people and many more homeless and without jobs \cite{deaths}. It is estimated that 14\% of the population will leave the country between 2017-2019 \cite{prcrisis}. Most people end up migrating to the United States, primarily because Puerto Ricans are born US citizens. In Figure \ref{fig:pr2usa}, we illustrate the post-Maria exodus to the United States. 
\begin{figure}[h]
\centering
    \includegraphics[width=\textwidth]{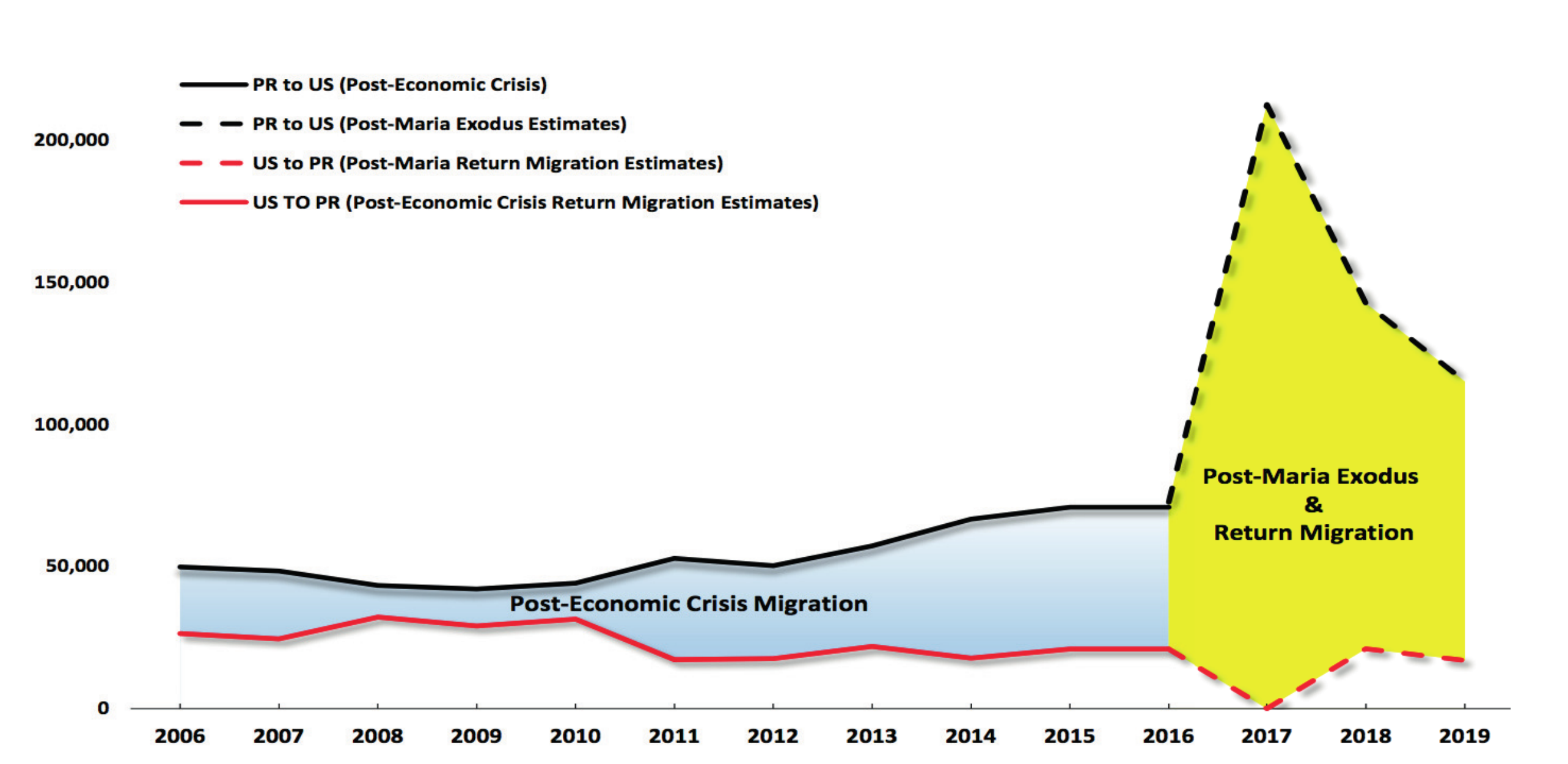}
    \caption{Post-Maria migration from Puerto Rico to the United States \cite{prcrisis,acs}. \label{fig:pr2usa}}
\end{figure}
The majority of these individuals that go to the United States move to the state of Florida \cite{prcrisis}. This particular state has all the proper characteristics for the {\it Aedes aegypti} to thrive, the mosquito that transmits diseases like dengue, chikungunya, and Zika \cite{cdc}. In Puerto Rico, these diseases have circulated in the past years \cite{cdcpr}, and infected individuals could potentially enter the United States and cause an outbreak in a population where these diseases have not caused problems previously.

In Venezuela the economic instability has caused millions of individuals to migrate to countries like Colombia, Peru, some parts of Central America, the United States and Europe. It is estimated that approximately 2.3 million Venezuelans now live abroad \cite{Vene,Vene2}. Another example is the Nicaraguan crisis because of government instability. In 2018, 24 000 people from Nicaragua have formally requested asylum in Costa Rica \cite{Nica}. The large number of people that are fleeing Nicaragua have caused cultural clashes, overcrowding in some areas and, most importantly, public health officials are under alert from imported malaria cases \cite{Nica2}.

More recently, in 2018 a large number of Honduran immigrants are attempting to enter the United States as refugees due to the economic and government crisis in Honduras \cite{HondCaravan}. These migrants are labeled as criminals and individuals who are likely to do harm to the native population \cite{criminalcricriminal}. However, it is far fetched to think that most of these migrants are criminals. Nevertheless, very few are considering the possible health risks that could arise due to these massive population movements.
Individuals may come with foreign pathogens that could cause health risks to the native population, as well as financial burden \cite{caravanPH}.

These recent events could lead to the introduction of infectious agents in a population, in turn causing major health concerns for local public health officials \cite{cdc2}. There are concerns over tuberculosis, foreign influenza strains, malaria, dengue, chikungunya, and Zika. The health of these migrant populations is crucial to the well being of the native population, since some of these pathogens could spread rapidly if the native population lacks immunity. For the very young and elderly, some of these diseases such as influenza and malaria could be fatal.

\subsection{United Sates immigration} \label{sec:datamex}
The United States has seen many changes on its immigrant population through the decades as illustrated in Figure \ref{fig:migusa}. The majority of individuals that attempt to enter the United States are looking for better job opportunities and a better way of life for their families, or trying to escape harsh political and economic conditions in their native countries. Moreover, many individuals that migrate to the United States do it legally or through the refugee program. However, recently the number of admissions has decreased \cite{mpi} (see Figure \ref{fig:refugee}).

\begin{figure}[!ht]
\centering
\subfloat[Change in the U.S. immigrant population by decade, 1860-2010.\label{fig:migusa}]{\includegraphics[width=0.49\textwidth]{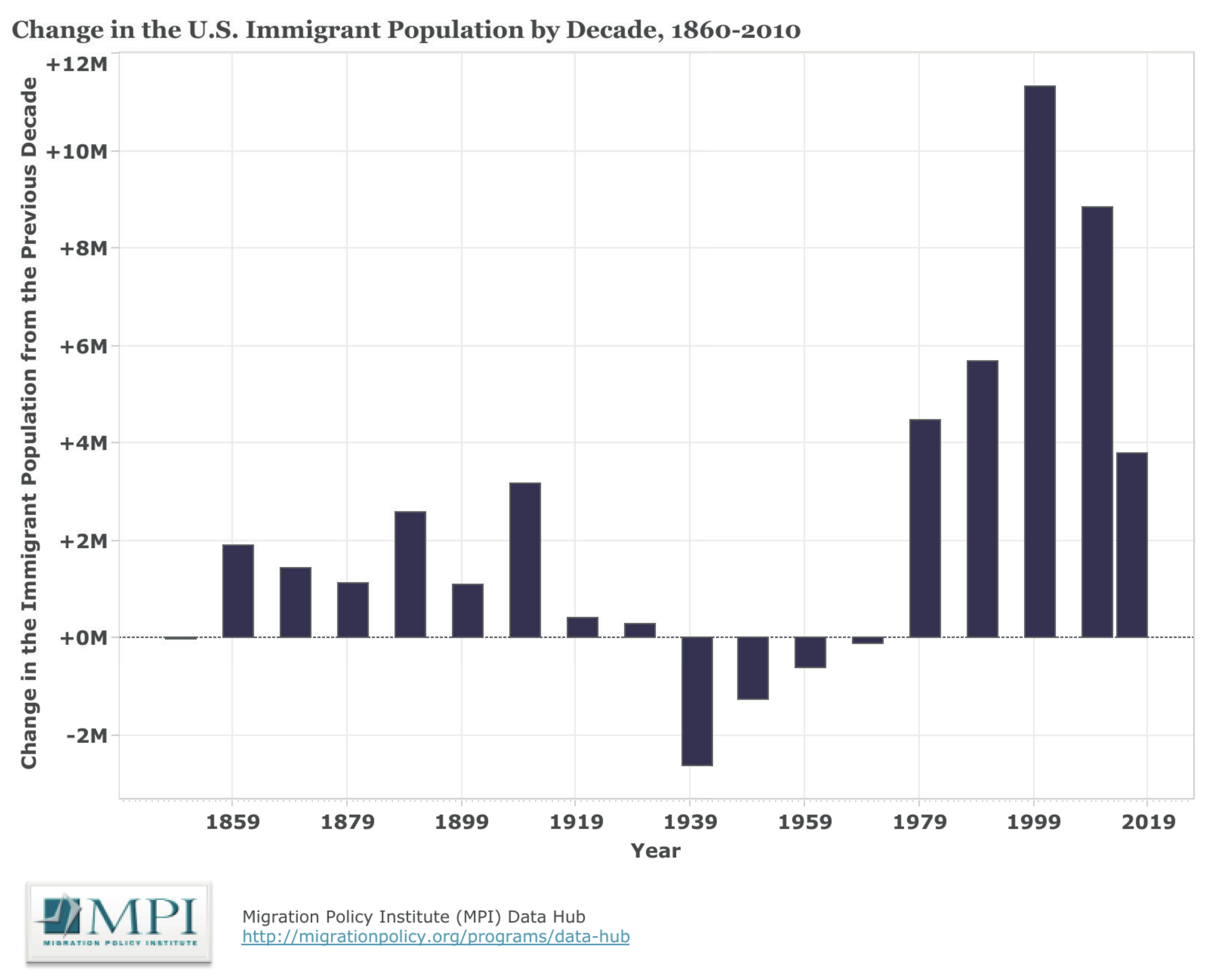}}\hfill
\subfloat[United States refugee admissions.\label{fig:refugee}]{\includegraphics[width=0.49\textwidth]{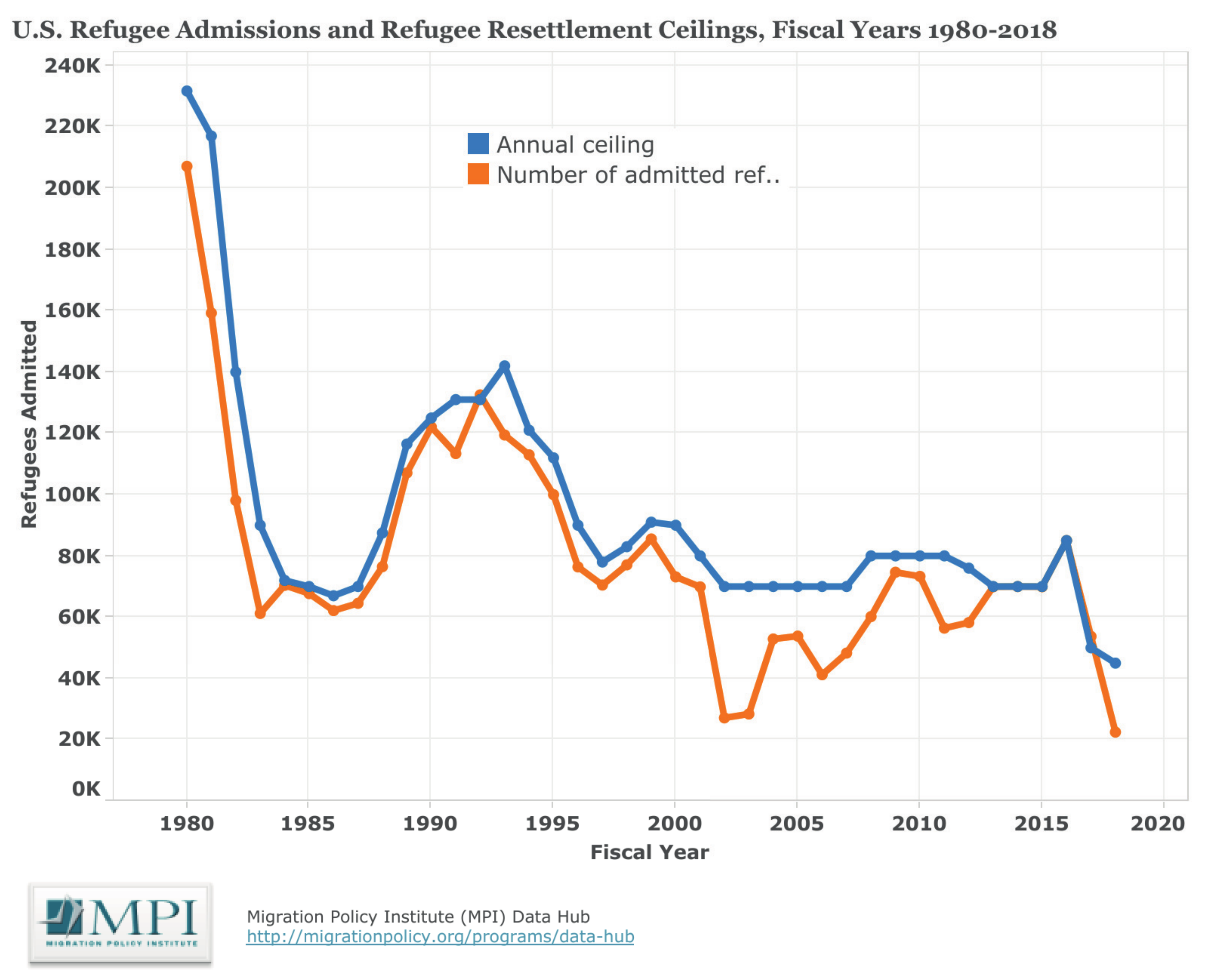}}\hfill
\subfloat[Mexican immigrant population from 1960-2016.\label{fig:mexusa}]{\includegraphics[width=0.49\textwidth]{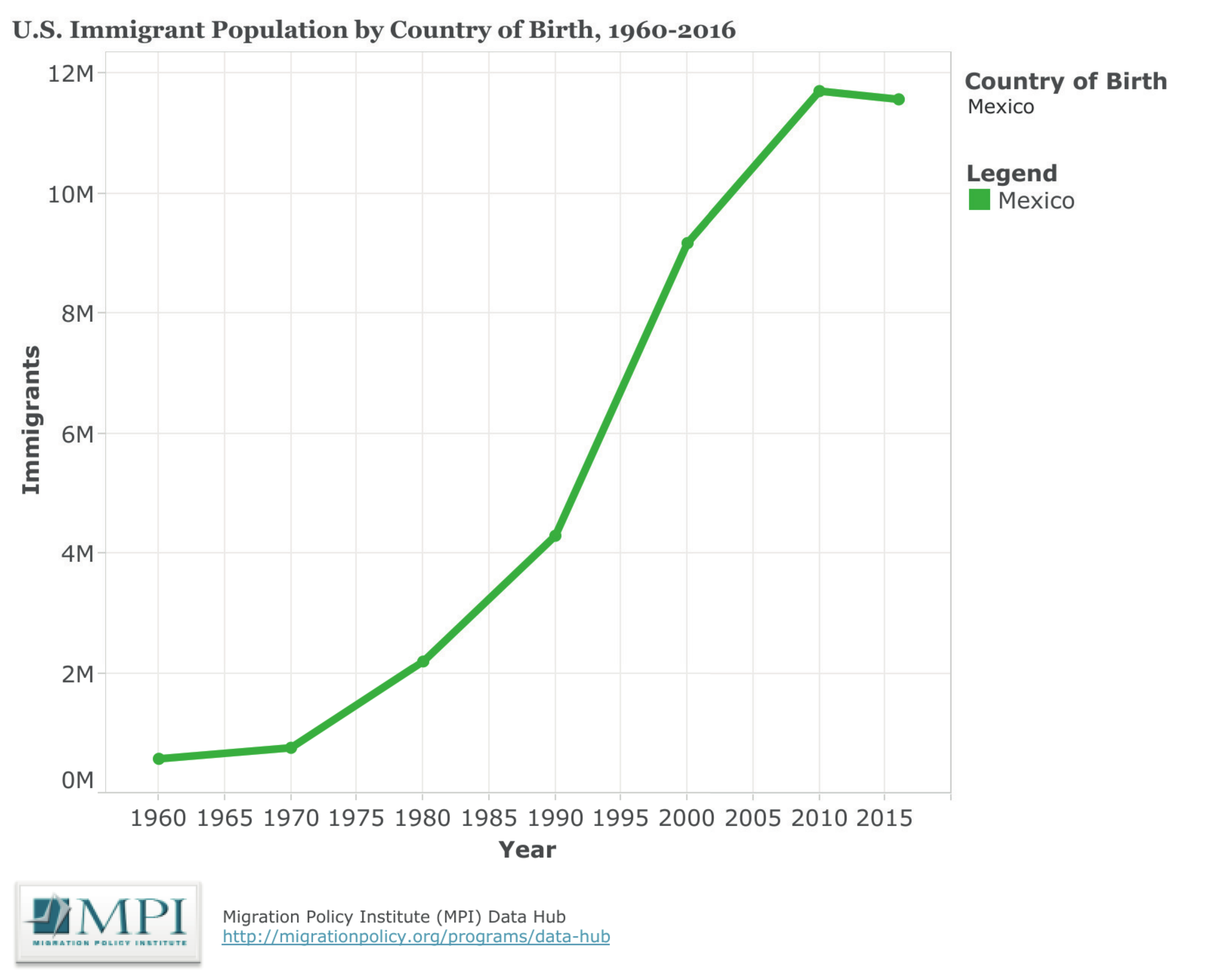}}\hfill
\subfloat[El Salvador, Honduras, Nicaragua and Venezuela immigrant population from 1960-2016.\label{fig:mig2usa}]{\includegraphics[width=0.49\textwidth]{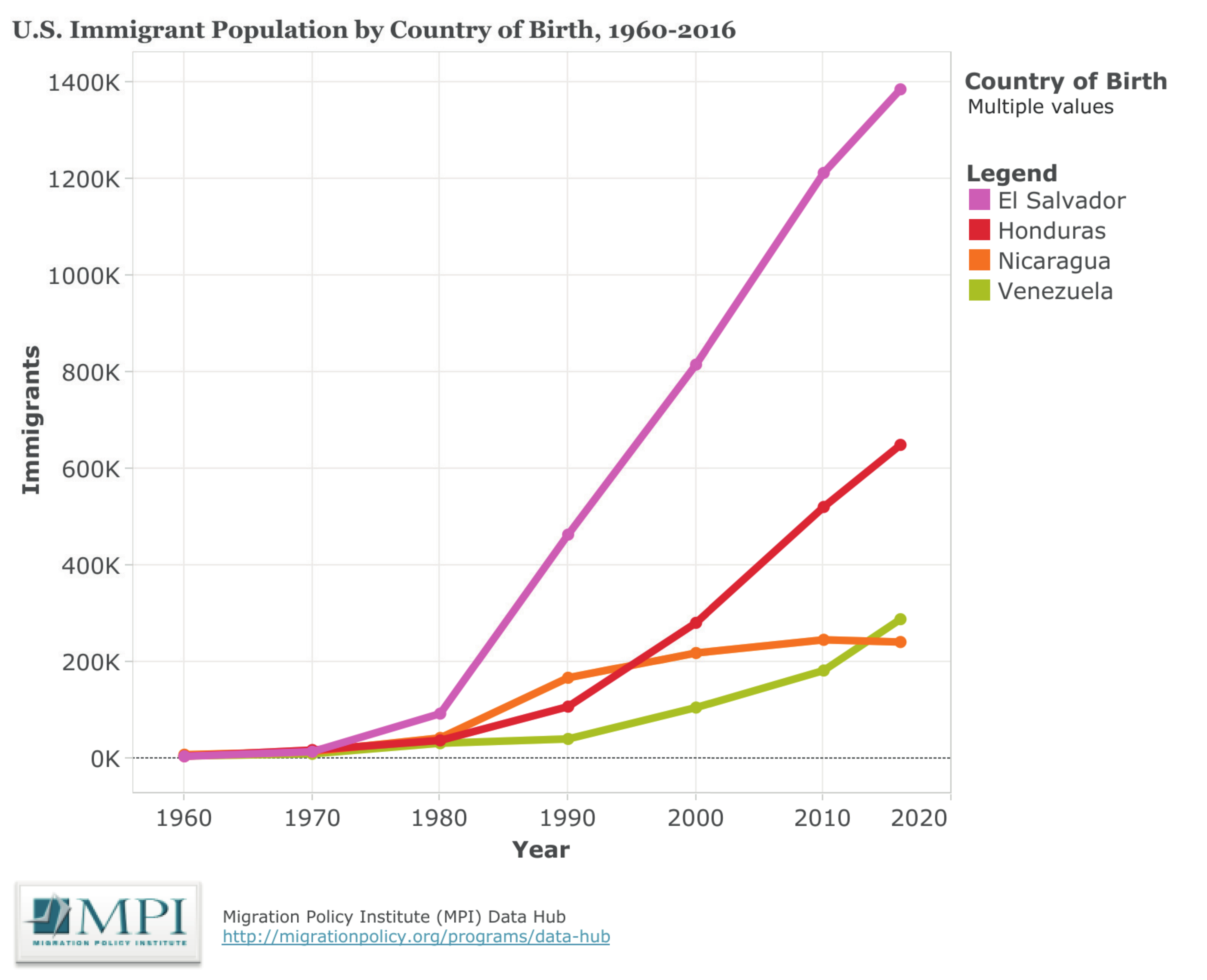}}
\caption{United States immigration and refugee data \cite{mpi}. \label{fig:usa}}
\end{figure} 

In Mexico, there is a long history of immigrants entering the United States both legally and illegally; see Figure \ref{fig:mexusa}. It is estimated that more than 6 million Mexican illegal immigrants currently live in the United States \cite{mpi}. Historically, most immigrants from Latin America come from Mexico. However, the landscape is changing as more immigrants from troubled countries are attempting to enter the United States; see Figure \ref{fig:mig2usa}.

In the last decades, immigrants have fled troubled countries looking for opportunities \cite{HondCaravan,Nica,Vene,Vene2,prcrisis}. Many come alone in order to establish themselves look for work and send money back to their families. Eventually, they try to bring their families or go back to their country, although this is likely a minority. As it pertains to the \lq\lq typical\rq\rq\ immigration patterns, immigrants have shown to positively contribute to society \cite{hirschman2013}. However, these large influx of individuals via {\it short-term immigration} can be disruptive if a new or absent infectious disease is introduced into a mostly susceptible population.

\section{A model with short-term immigration} \label{sec:model}
We explore an age-dependent nonlinear partial differential equation SIR model with {\it short-term immigration} effects, where $S(t,a)$, $I(t,a)$ and $R(t,a)$ represent susceptible, infected and recovered individuals at time $t$ and age $a$, respectively. Susceptible individuals can become infected by having contact with an infectious individual at rate $\beta(a) B(t,a)$, where $B(t,a)$ is the probability that a person contacted by a susceptible of age $a$ is infectious and $\beta(a)$ is the transmission rate at age $a$; they can also exit the system at rate $\mu(a)$. Infected individuals can recover at rate $\gamma(a)$ or exit the system at rate $\mu(a)$.

{\it Short-term immigration} is modeled via $\nu(t,a)$, the population of immigrants of age $a$ entering the new population per unit time ($t$). We assume that a fixed percentage of the immigrants are susceptible, infected and recovered, and define $\rho_S \nu(t,a)$, $\rho_I \nu(t,a)$, $\rho_R \nu(t,a)$ the proportions of susceptible, infected and recovered immigrants, respectively; see Figure \ref{fig:model}. We assume that  
$\rho_S, \rho_I, \rho_R\geq 0$, $\rho_S+\rho_I+\rho_R=1$, and that $\nu(t,a)$ is non-zero only at some finite time interval $t\in [0,T]$, since we are interested in a {\it short-term immigration}.
\begin{figure}[!t]
	\centering
	\includegraphics[width=0.55\textwidth]{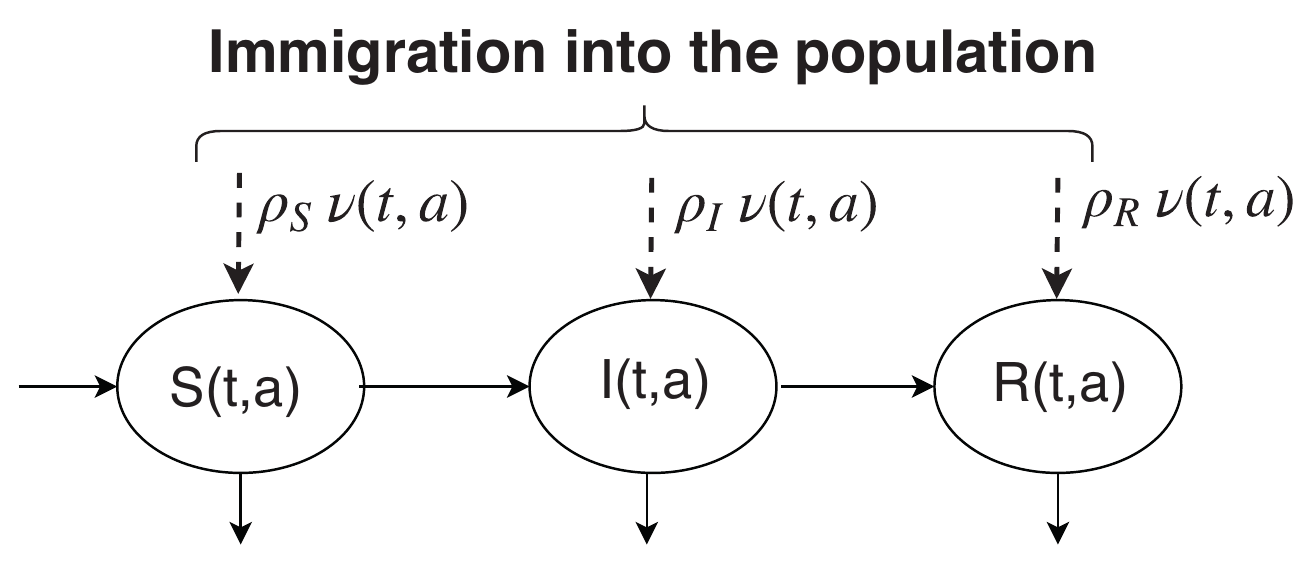}
      \caption{Model diagram highlighting immigration. \label{fig:model}}
\end{figure}
The age-dependent 
contact structure is modeled via $p(t,a,a')$, the mixing density, which gives the proportion of contact that individuals of age $a$ have with individuals of age $a'$ given that they had contact with somebody at time $t$. We will restrict ourselves to the case of proportional mixing; this is, $$p(t,a,a') \equiv p(t,a) =  \frac{c(a)n(t,a)}{\int_0^\infty c(a) n(t,a)\ da}, $$ with $c(a)$ the age-specific per-capita contact/activity rate; see, e.g., \cite{shim2006}. We then define the force of infection $$B(t) = \int_0^\infty \frac{I(t,a)}{n(t,a)} p(t,a)da,$$ where $n(t,a):= S(t,a)+I(t,a)+R(t,a)$ is the total population.

The model we just described is given by the system of partial differential equations
\begin{align}\label{PDE}
\left( \frac{\partial}{\partial t} + \frac{\partial}{\partial a} \right) S(t,a) &= \rho_S\ \nu(t,a) - \beta(a) S(t,a) B(t) -\mu(a)S(t,a),\\
\left( \frac{\partial}{\partial t} + \frac{\partial}{\partial a} \right) I(t,a) &= \rho_I\ \nu(t,a) + \beta(a) S(t,a) B(t) -(\mu(a) + \gamma(a)) I(t,a),\nonumber\\
\left( \frac{\partial}{\partial t} + \frac{\partial}{\partial a} \right) R(t,a) &=  \rho_R\ \nu(t,a) + \gamma(a) I(t,a) - \mu(a) R(t,a),\nonumber
\end{align} 
for $t>0$, $a>0$, and boundary conditions for $t=0$ and $a=0$ given by 
\begin{align} \label{pde_ic}
S(t,0)&=\int_0^\infty f(a) n(t,a)\ da , & I(t,0)&=0, & R(t,0)&=0,\nonumber\\
S(0,a)&=S_0(a), & I(0,a)&=I_0(a), & R(0,a)&=R_0(a),
\end{align}
where $f(a)$ is the fertility function per person of age a, and $S_0(a), I_0(a), R_0(a)$ are given initial conditions. 

We remark that the total population $n(t,a)$ satisfies the initial value problem
\begin{align}\label{eq:nta}
\left( \frac{\partial}{\partial t} +\frac{\partial}{\partial a} \right) n(t,a) &= \nu(t,a) -\mu(a) n(t,a),\\
n(t,0) &= \int_0^\infty f(a) n(t,a)\ da =: G(t),\nonumber\\
n(0,a) &= S_0(a)+I_0(a)+R_0(a) =: n_0(a),\nonumber
\end{align}
with $t>0$ and $a>0$. This model for $n(t,a)$ is similar to the one considered in \cite{franceschetti2008}, for which permanent immigration $\nu = \nu(a)$ is considered, independent of $t$.


\section{Immigration-free and Infection-free non-uniform steady state distributions} \label{sec:freesteadyState}
In this section we explore the {\it immigration-free} and  {\it infection-free} steady state distribution, as its local and global stability, as well as the {\it basic reproductive number}, based on \cite{sanchez2018x}.

In \cite{sanchez2018x}, an age-structured model with nonlinear recidivism is studied. Even though \cite{sanchez2018x} does not include immigration, our model will have a similar steady-state distribution as theirs, since we are assuming that $\nu(t,a)$ is non-zero for $t\in [0,T]$ for relatively small $t$, so immigration is temporary (note also that in our case there is no subsequent infections for the recovered class).

We can then define the {\it basic reproductive number}, $\mathcal{R}_0$, by
\begin{equation*} 
\mathcal{R}_0 := \int_0^\infty \int_0^a  p_\infty (a) \beta(h) e^{- \int_h^a \gamma(k)dk}  dh da,
\end{equation*}
where $$p_\infty (a) =  \frac{c(a)n^*(a)}{\int_0^\infty c(a) n^*(a)\ da}, $$ and $n^*(a)$ is the steady-state population; see \cite[Section 2]{sanchez2018x}. We have the two following results:

\begin{theorem}
If $\mathcal{R}_0 < 1$, the immigration-free and infection-free non-uniform steady state distributions of System \eqref{PDE} is locally asymptotically stable and unstable if $\mathcal{R}_0>1$. \label{thm:global}
\end{theorem}
\begin{proof} 
See \cite[Theorem 3.1]{sanchez2018x}.
\end{proof}

We also have the following result about global stability:
\begin{theorem}\label{th:globalStab}
Assume that the basic reproductive number satisfies $\mathcal{R}_0<1$. Then, the solution of System \eqref{PDE} is globally asymptotically stable.
\end{theorem}
\begin{proof}
The result follows from \cite[Theorem 3.2]{sanchez2018x}, for the particular case when there are no subsequent infections.
\end{proof}

\section{Numerical experiments} \label{sec:exp}
In this section, we first describe the numerical implementation for solving System \eqref{PDE}. We then implement distinct numerical scenarios based on recent events related to large immigrant groups \cite{HondCaravan,Nica,Vene,Vene2,prcrisis}. All parameters are in units of $year^{-1}$. Particularly, we choose $\nu(t,a)$ based on the number of Mexican immigrants that entered the United States in 2016, and $n_0(a)$, $f(a)$, $\mu(a)$ were chosen according to US demographic data taken from \cite{mpi}. 

\subsection{Numerical implementation} \label{sec:numer_imp}
There is numerical evidence that the numerical scheme is sensitive to the computation of $B(t)$, for which we need to evaluate $n(t,a)$ at each time $t$. By using the method of characteristics, from system \eqref{eq:nta} we obtain
\begin{equation} \label{eq:nta_sol}
    n(t,a) = \left\lbrace
    \begin{array}{rl}
        \pi(a) \left( \dfrac{n_0(a-t)}{\pi(a-t)}+ \displaystyle\int_{a-t}^a \dfrac{\nu(\tau+t-a,\tau)}{\pi(\tau)}\ d\tau \right)& {\rm if }\ t\le a,  \\
        \pi(a)\left( G(t-a)+ \displaystyle\int_0^a \dfrac{\nu(\tau+t-a,\tau)}{\pi(\tau)}\ d\tau \right) & {\rm if }\ t>a,
    \end{array}
    \right.
\end{equation}
where $$\pi(a) := e^{-\int_0^a \mu(h)\ dh}$$ is the proportion of individuals that survive at age $a$.
We note that the model requires the compatibility condition $$n_0(0) = \int_0^\infty f(a) n_0(a)\ da$$ in order to obtain a continuous solution $n(t,a)$.
Since $G(t) = \int_0^\infty f(a) n(t,a)\ da $, we can substitute $n(t,a)$ as given in \eqref{eq:nta_sol} to obtain the Volterra integral equation
\begin{equation*}
G(t) = \int_0^t G(t-a) K(a)\ da + H(t),
\end{equation*}
with kernel $K(a) = f(a)\pi(a)$, and non-homogeneous term 
\begin{align*}
H(t) =&\displaystyle \int_0^t \int_0^a f(a) \dfrac{\pi(a)}{\pi(\tau)}\nu(\tau+t-a,\tau) d\tau da\  \\
&+\displaystyle \int_t^\infty \int_{a-t}^a f(a) \dfrac{\pi(a)}{\pi(\tau)}\nu(\tau+t-a,\tau) d\tau da  \\ &+\displaystyle \int_t^\infty f(a) \dfrac{\pi(a)}{\pi(a-t)}n_0(a-t) da.
\end{align*}

Nevertheless, for {\it short-term immigration}, we are interested in the behavior of the model for small values of $t$. In that sense, we have that $$G(t) = H(t),$$ since the kernel $K(t,a)$ vanishes for small values of $t$ due to the fertility function $f$. We can then compute accurately $G(t)$, which allows us to obtain $n(t,a)$ from \eqref{eq:nta_sol}.


We then discretize \eqref{PDE} with a first-order upwind finite difference scheme and approximate its solution on the rectangle $$\lbrace (t,a)\in [0,T]\times [0,A]\rbrace,$$ where $A$ represents the maximum human lifespan. We first construct a uniform grid with equidistant points. Consider a partition $\lbrace t_j \rbrace_{j=0}^{N_T}$ of the interval $[0,T]$ with $N_T$ equidistant points, and a partition $\lbrace a_k \rbrace_{k=0}^{N_A}$ of the interval $[0,A]$ with $N_A$ equidistant points. Thus, the nodes $(t_j,a_k)$ on the rectangular mesh are given by 
$$(t_j,a_k) = \left( j\Delta t, k\Delta a\right), \quad j\in\lbrace 0,1,\ldots,N_T\rbrace,\quad  k\in\lbrace 0,1,\ldots,N_A\rbrace,$$
where 
$$\Delta t := \frac{T}{N_T-1},\quad \Delta a := \frac{A}{N_A-1}$$ are the corresponding step sizes. We recall that it is necessary that $\Delta t <\Delta a$ in order to satisfy the CFL and stability conditions of the scheme.

For any function $X$ and a grid point $(t_j,a_k)$, we denote the approximation of $X(t_j,a_k)$ by $X_k^j$. 
Since
\begin{align*}
\left( \frac{\partial}{\partial t} + \frac{\partial}{\partial a} \right) X(t_j,a_k) =&\  \frac{X(t_j+\Delta t, a_k)-X(t_j,a_k)}{\Delta t} + \mathcal{O}(\Delta t) \\
& \quad +  \frac{X(t_j,a_k)- X(t_j,a_{k}-\Delta a)}{\Delta a} + \mathcal{O}(\Delta a),
\end{align*}
we approximate the derivatives by
\begin{equation*}
\left( \frac{\partial}{\partial t} + \frac{\partial}{\partial a} \right) X(t_j,a_k) \approx  \frac{X^{j+1}_{k} - X^{j}_{k}}{\Delta t} + \frac{X^j_{k} - X^{j}_{k-1}}{\Delta a}.
\end{equation*}
Similarly, define $\nu_k^j = \nu(t_j,a_k)$, $\beta_k:=\beta(a_k)$, $\mu_k:=\mu(a_k)$, $\gamma_k:=\gamma(a_k)$ and $B^j:=B(t_j)$. Then, by evaluating at all the grid points, the discretization for system \eqref{PDE} is given by the explicit system
\begin{align} \label{PDE:discrete}
\frac{S^{j+1}_{k} - S^{j}_k}{\Delta t} + \frac{S^j_k - S^{j}_{k-1}}{\Delta a} &= \rho_S \nu_k^j - \beta_k S^j_k B^j - \mu_k S^j_k,\nonumber\\
\frac{I^{j+1}_k - I^{j}_k}{\Delta t} + \frac{I^j_k - I^{j}_{k-1}}{\Delta a} &= \rho_I \nu_k^j + \beta_k S^j_k B^j - (\mu_k+\gamma_k) I^j_k, \nonumber \\
\frac{R^{j+1}_k - R^{j}_k}{\Delta t} + \frac{R^j_k - R^{j}_{k-1}}{\Delta a} &= \rho_R \nu_k^j + \gamma_k I^j_k- \mu_k R_k^j,
\end{align}
for $1\leq k\leq N_A$ and $0\leq j\leq N_T-1$. We recall that the initial conditions \eqref{pde_ic} provide the values $S_0^j$, $S_k^0$, $I_0^j$, $I_k^0$, $R_0^j$, and $R_k^0$. The integral $$B^j = B(t_j) = \int_0^\infty \dfrac{I(t_j,a)}{n(t_j,a)} p(t_j,a)\ da = \dfrac{\int_0^\infty c(a) I(t_j,a)\ da}{\int_0^\infty c(a) n(t_j,a)\ da}  $$ is approximated via MATLAB's command \texttt{integral}. We remark that for a fixed time $t_j$, $B^j$ depends only on values of $I$ at the same time $t_j$. Finally, solving for $S^{j+1}_k$, $I^{j+1}_k$ and $R^{j+1}_k$ from \eqref{PDE:discrete}, we obtain that
\begin{subequations}
\begin{align*}
\begin{split} 
S^{j+1}_k  &= S^{j}_k + \Delta t \left( \rho_S \nu_k^j - \beta_k S^j_k B^j - \mu_k S^j_k  - \frac{S^j_k - S^{j}_{k-1}}{\Delta a}\right),
\end{split}\\
\begin{split}
I^{j+1}_k &= I^{j}_k +  \Delta t\left( \rho_I \nu_k^j + \beta_k S^j_k B^j - (\mu_k+\gamma_k) I^j_k  -  \frac{I^j_k - I^{j}_{k-1}}{\Delta a} \right),
\end{split}\\
\begin{split} 
R^{j+1}_k &= R^{j}_k+ \Delta t\left( \rho_R \nu_k^j + \gamma_k I^j_k- \mu_k R_k^j -  \frac{R^j_k - R^{j}_{k-1}}{\Delta a} \right),
\end{split}
\end{align*}
\end{subequations}
for $1\leq k\leq N_A$ and $0\leq j\leq N_T-1$. Therefore, starting with the initial conditions $S_0^j$, $S_k^0$, $I_0^j$, $I_k^0$, $R_0^j$, $R_k^0$, we can compute, in successive time steps, the values of the unknowns on the grid points.

\subsection{Numerical simulations} \label{ex:var_coef}
We discuss different scenarios in order to study the effect of {\it short-term immigration} in a population. The mortality rate $\mu(a)$, the initial population $n_0(a)$ and the fertility rate are based on data for the US from \cite{mpi}. The proportion of immigrants $\nu(t,a)$ is assumed to be constant as a function of $t$ in a two-year period $0\leq t\leq 2$, based on the age distribution of immigrants from Mexico into the US on 2015 \cite{mpi}. These functions are shown in Figure \ref{fig:param_a}.

\begin{figure}[H]
\centering
\subfloat[Natural mortality rate, $\mu(a)$.\label{fig:mu}]{\includegraphics[width=0.4\textwidth]{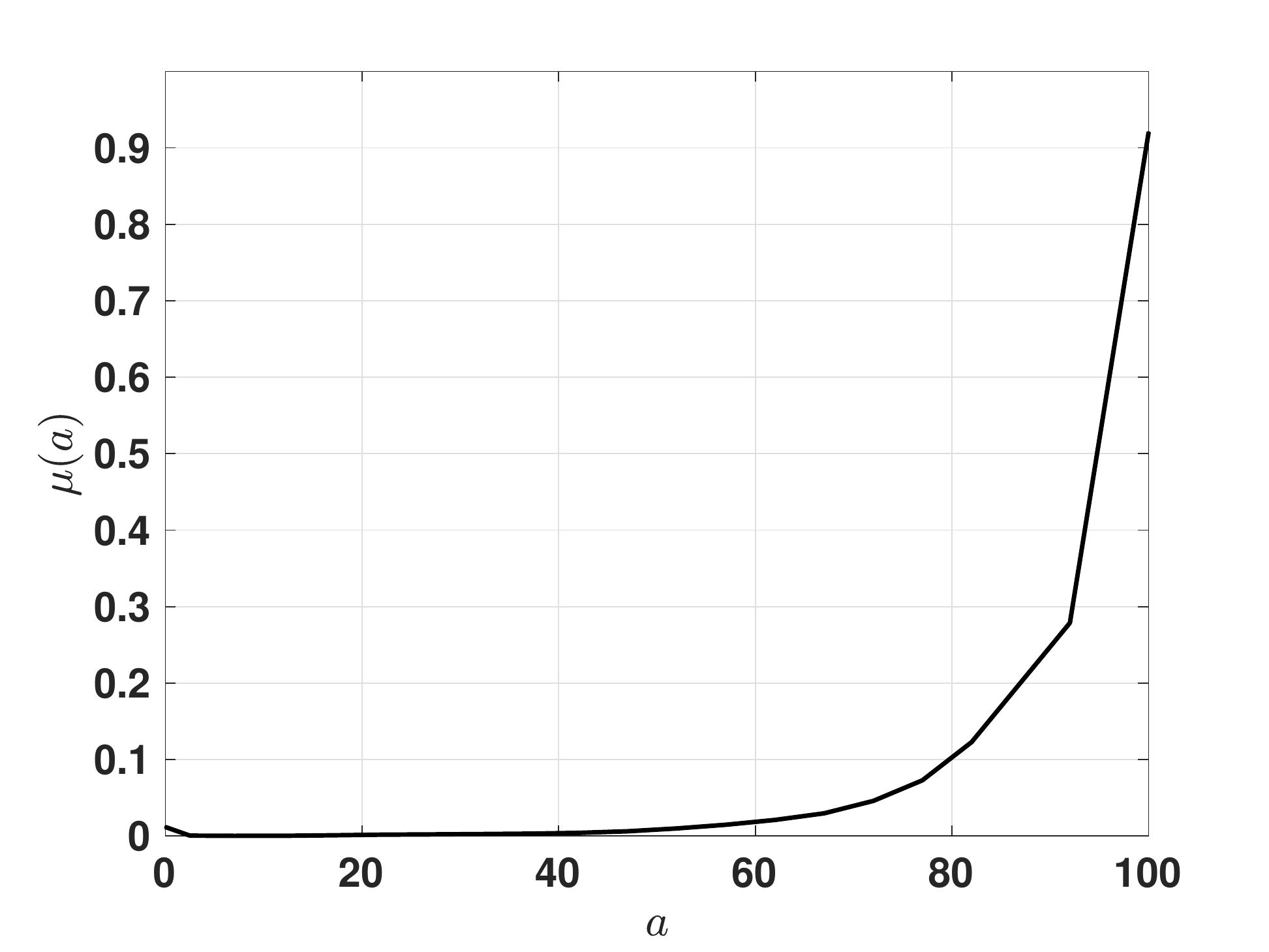}}\hfill
\subfloat[Fertility function, $f(a)$.\label{fig:fert}]{\includegraphics[width=0.4\textwidth]{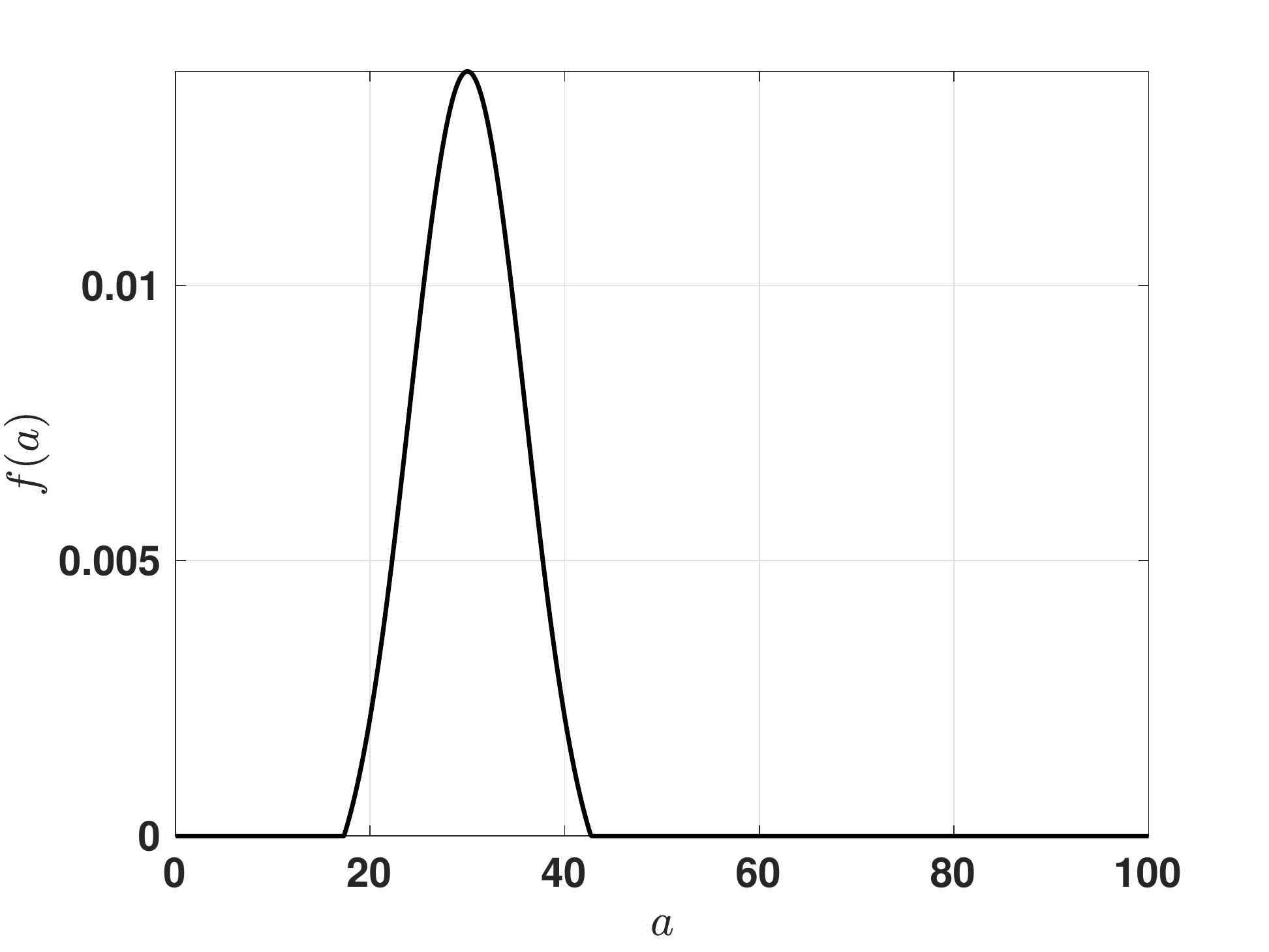}}\hfill
\subfloat[Initial population, $n_0(a)$.\label{fig:n0}]{\includegraphics[width=0.4\textwidth]{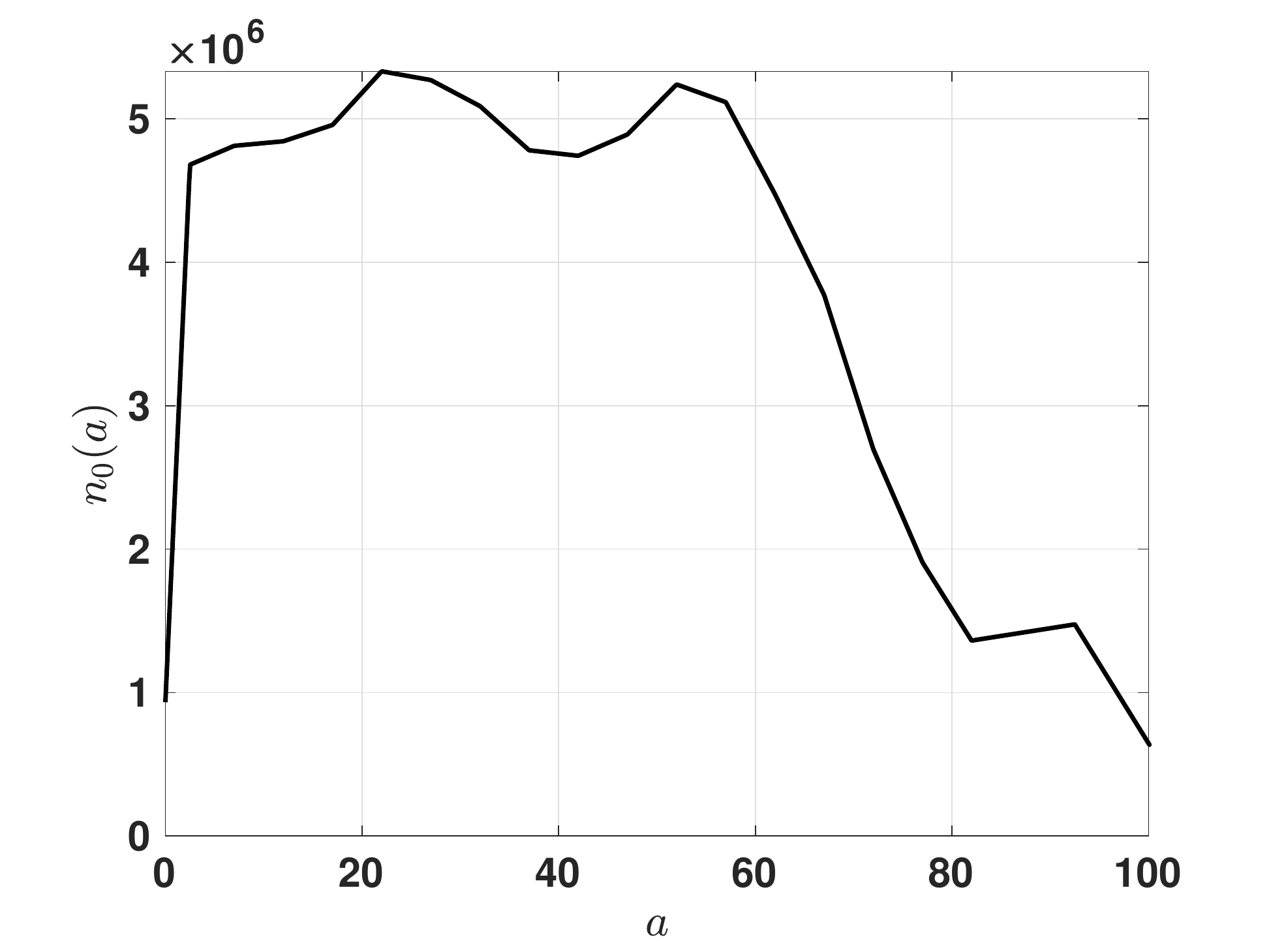}}\hfill
\subfloat[Immigrant population of age $a$ entering the new population, $\nu(t,a)$.\label{fig:nu}]{\includegraphics[width=0.4\textwidth]{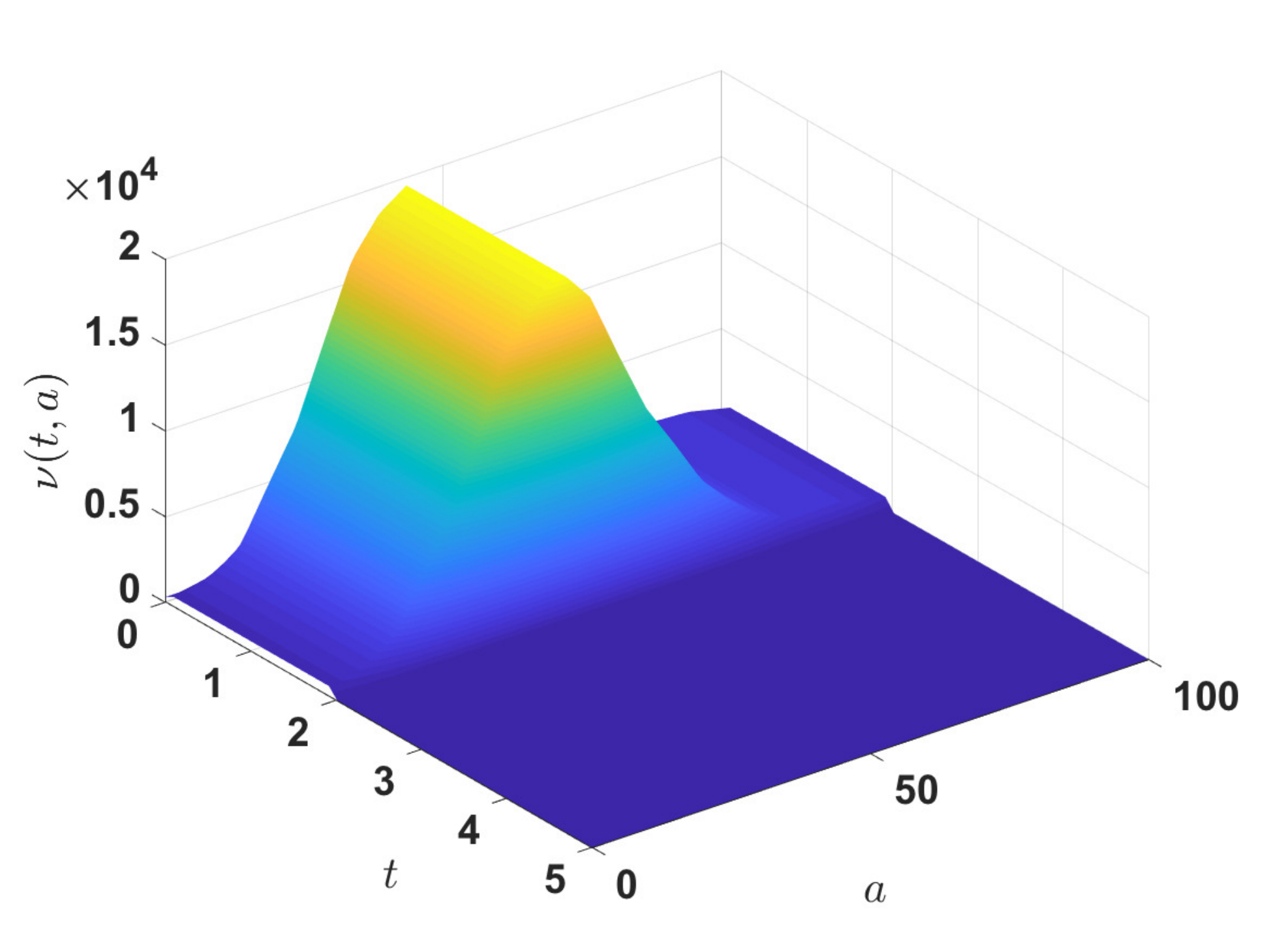}}
\caption{Distributions used on the numerical experiments based on \cite{mpi}.\label{fig:param_a}}
\end{figure}

\begin{figure}[H]
\centering
\subfloat[$\gamma(a)$.]{\includegraphics[width=0.4\textwidth]{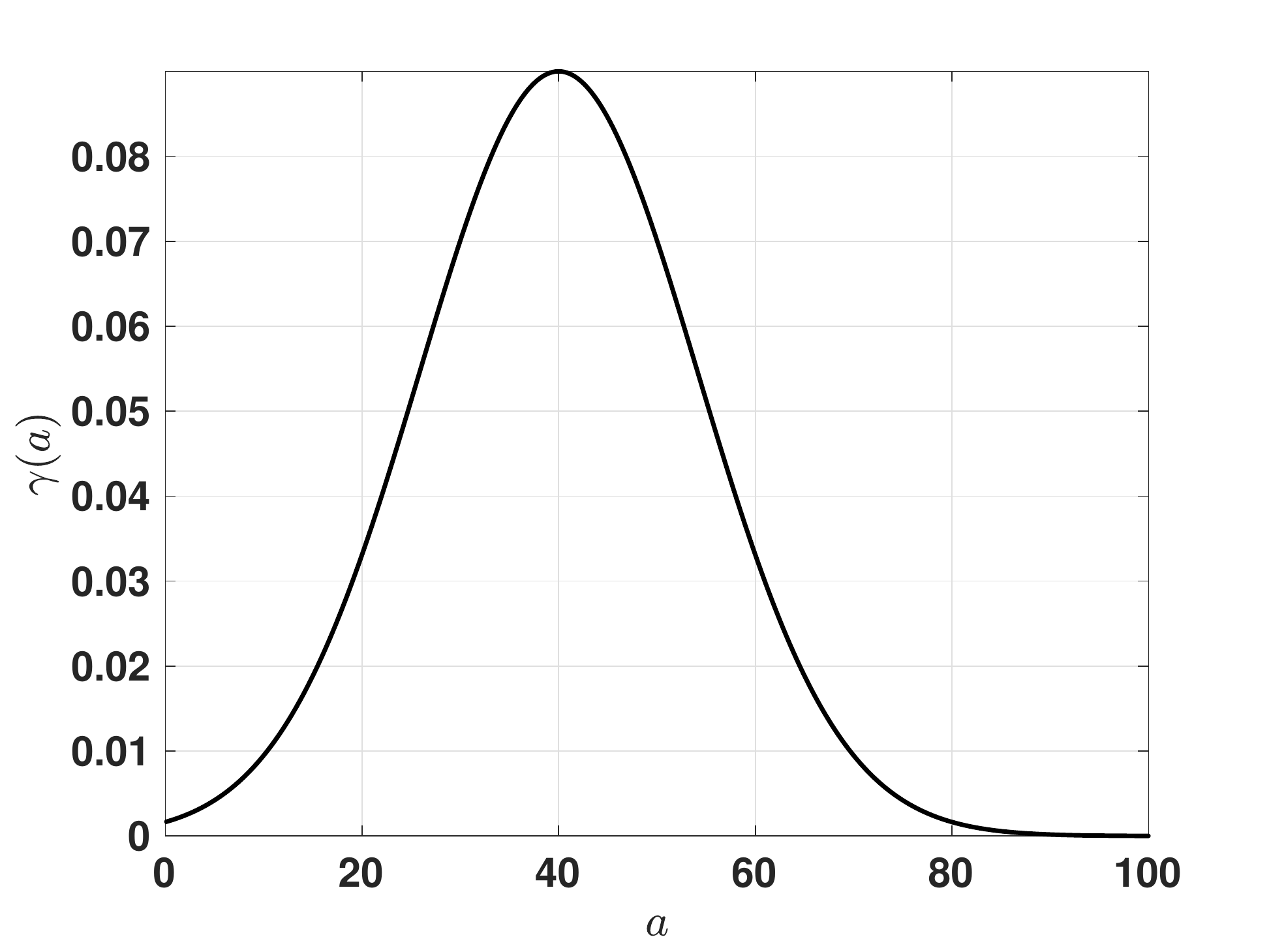}}\hfill
\subfloat[$c(a)$.]{\includegraphics[width=0.4\textwidth]{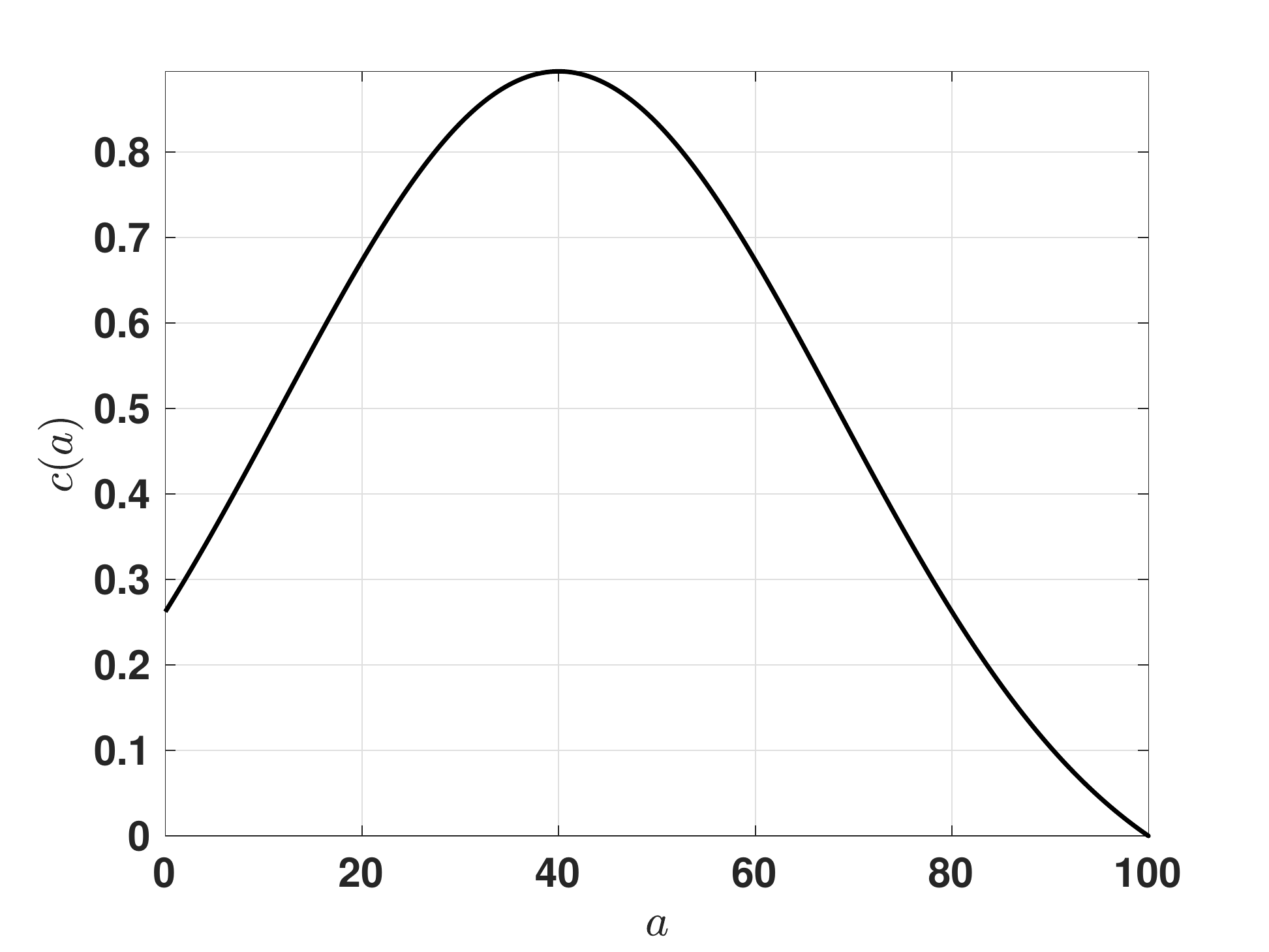}}\hfill
\caption{Recovery rate (left) and contact rate (right) as functions of age, modeled via a normal distribution.\label{fig:ex1_gamma_c}}
\end{figure}

\begin{example} \label{ex1} 
{\rm
Consider a susceptible population ($I_0(a) \equiv 0 \equiv R_0(a)$) that receives some infected immigrants. In this experiment, we use the parameters $c(a)$, $\gamma(a)$ as shown in Figure \ref{fig:ex1_gamma_c}, $\rho_R=0$ and vary $\rho_I\in [0,1]$, $\rho_S = 1-\rho_I$. 
\begin{figure}[!t]
\centering
\subfloat[$\beta_1(a)$.\label{fig:beta1}]{\includegraphics[width=0.45\textwidth]{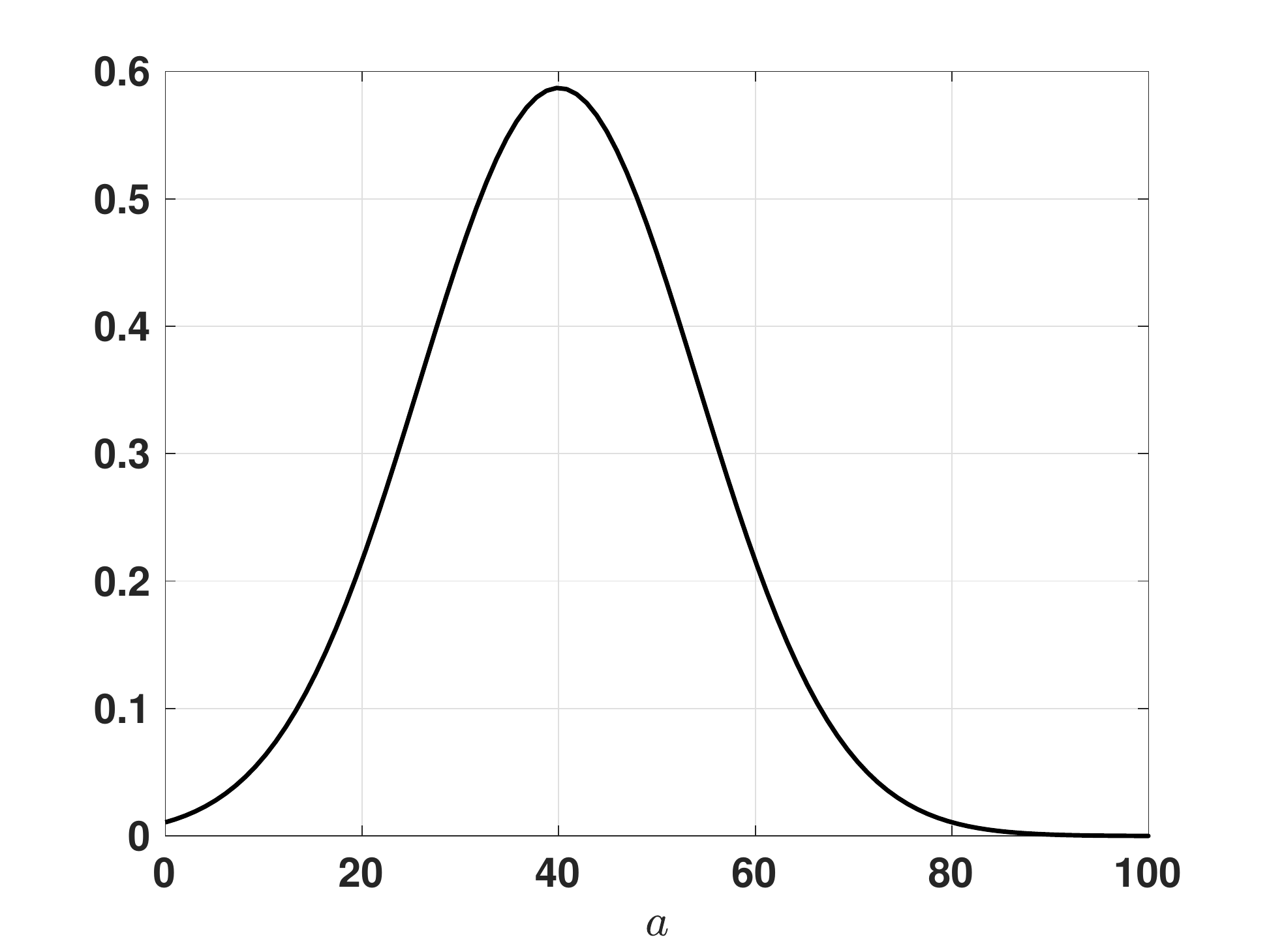}}\hfill
\subfloat[${I(10,a)}/{n(10,a)}$ for $\beta_1(a)$, $\rho_R=0$.\label{fig:ex1b1}]{\includegraphics[width=0.45\textwidth]{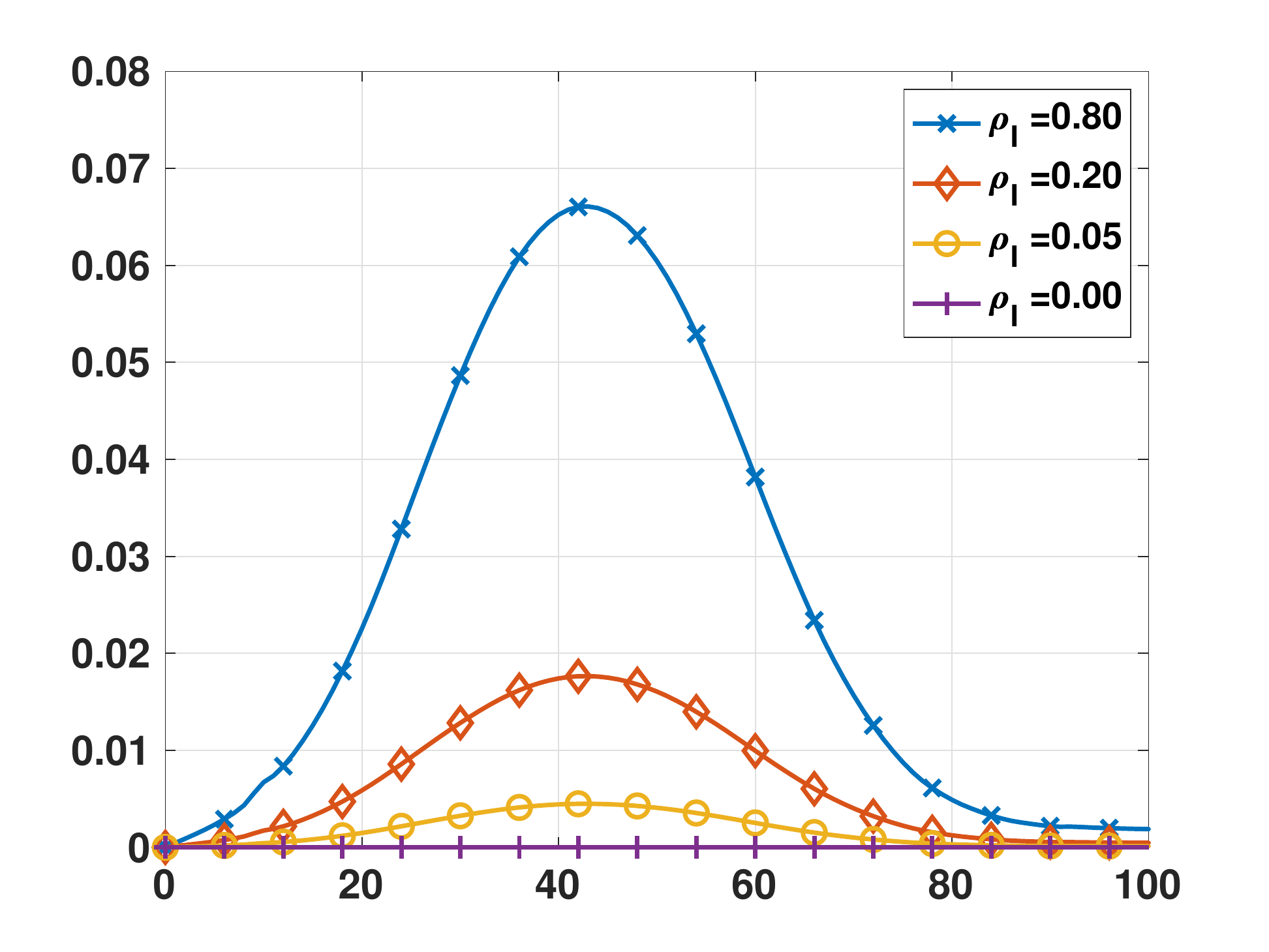}}\hfill
\subfloat[$\beta_2(a)$.\label{fig:beta2}]{\includegraphics[width=0.45\textwidth]{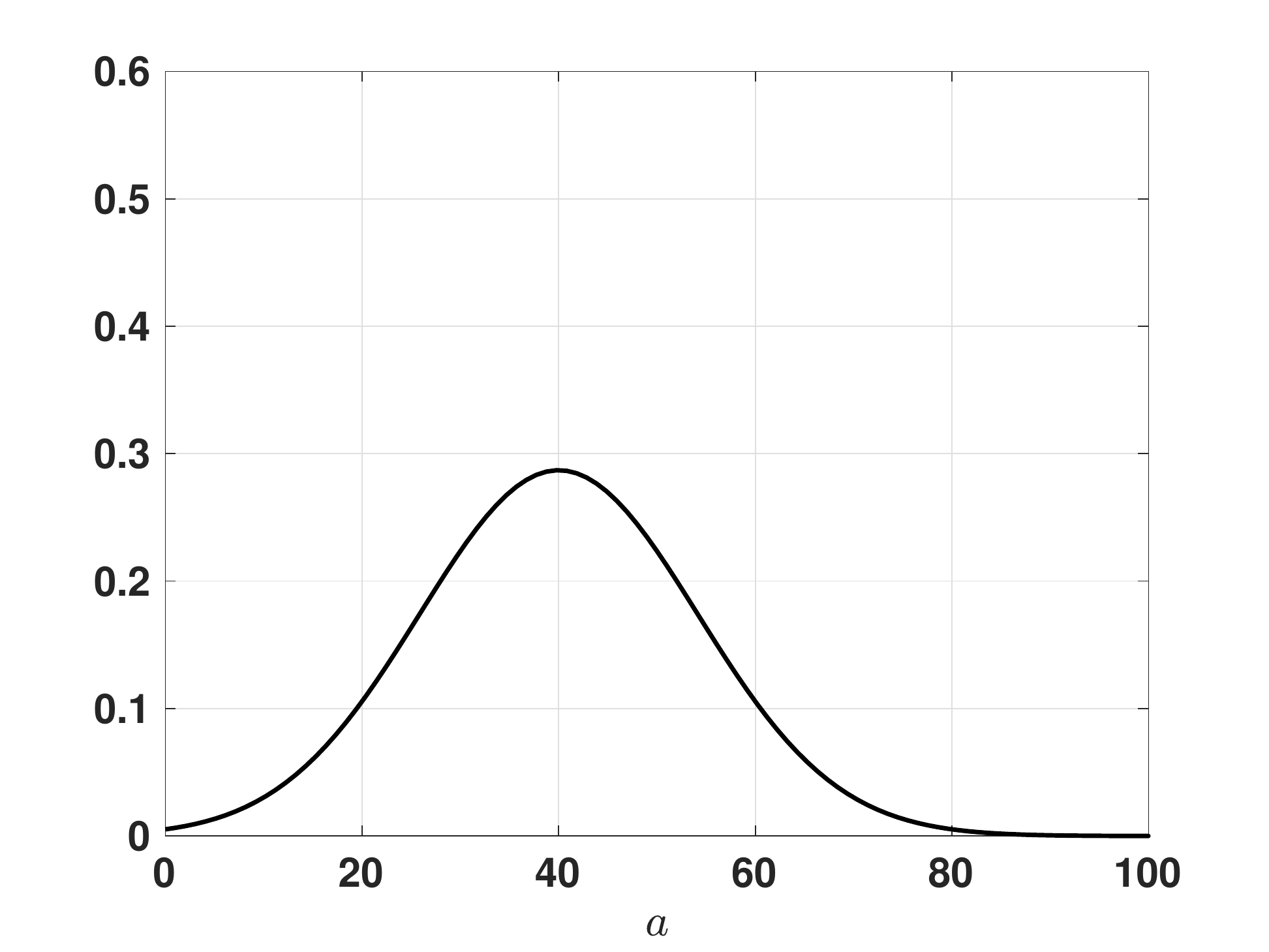}}\hfill
\subfloat[${I(10,a)}/{n(10,a)}$ for $\beta_2(a)$, $\rho_R=0$. \label{fig:ex1b2}]{\includegraphics[width=0.45\textwidth]{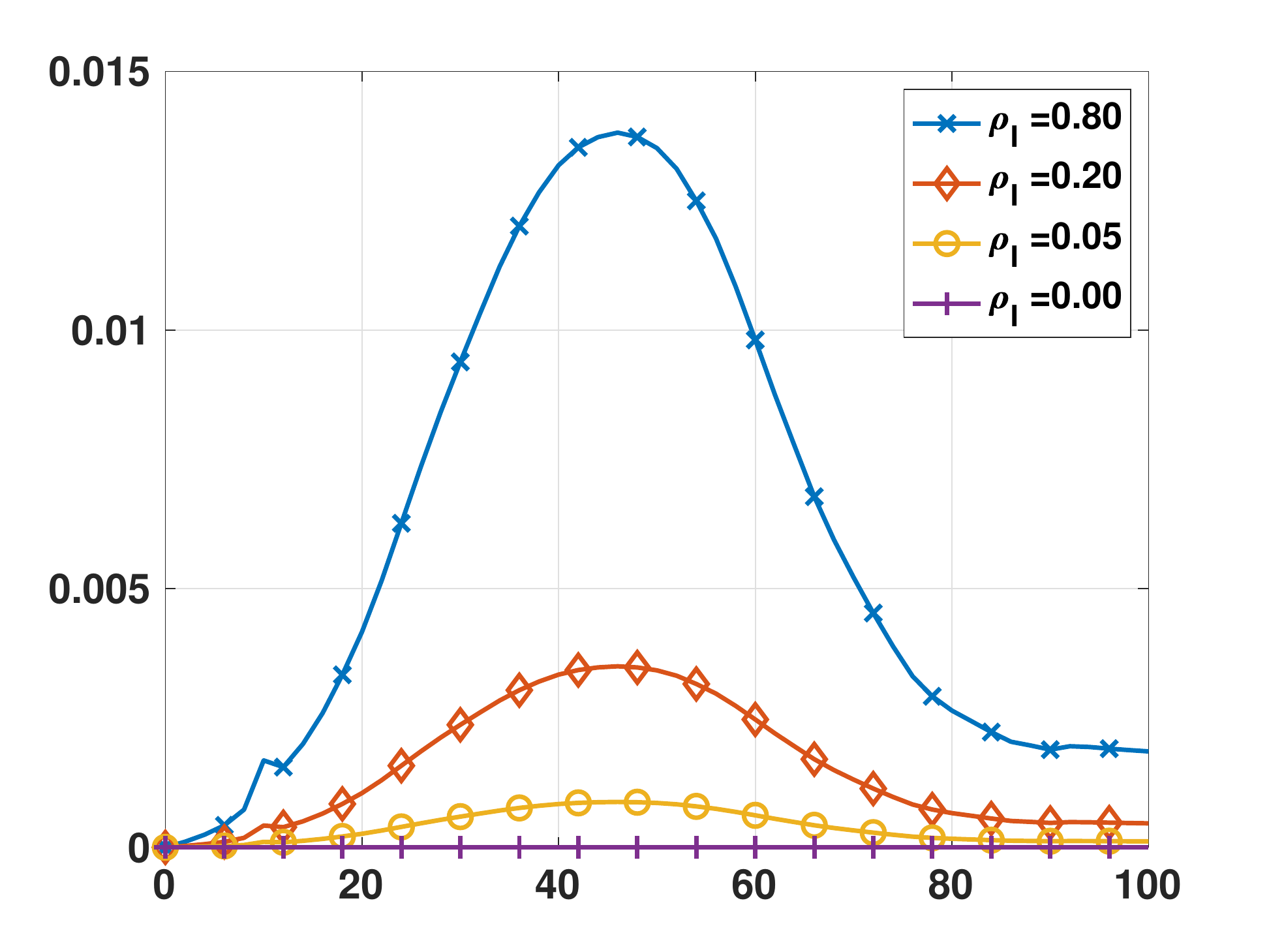}}\hfill
\subfloat[$\beta_3(a)$.\label{fig:beta3}]{\includegraphics[width=0.45\textwidth]{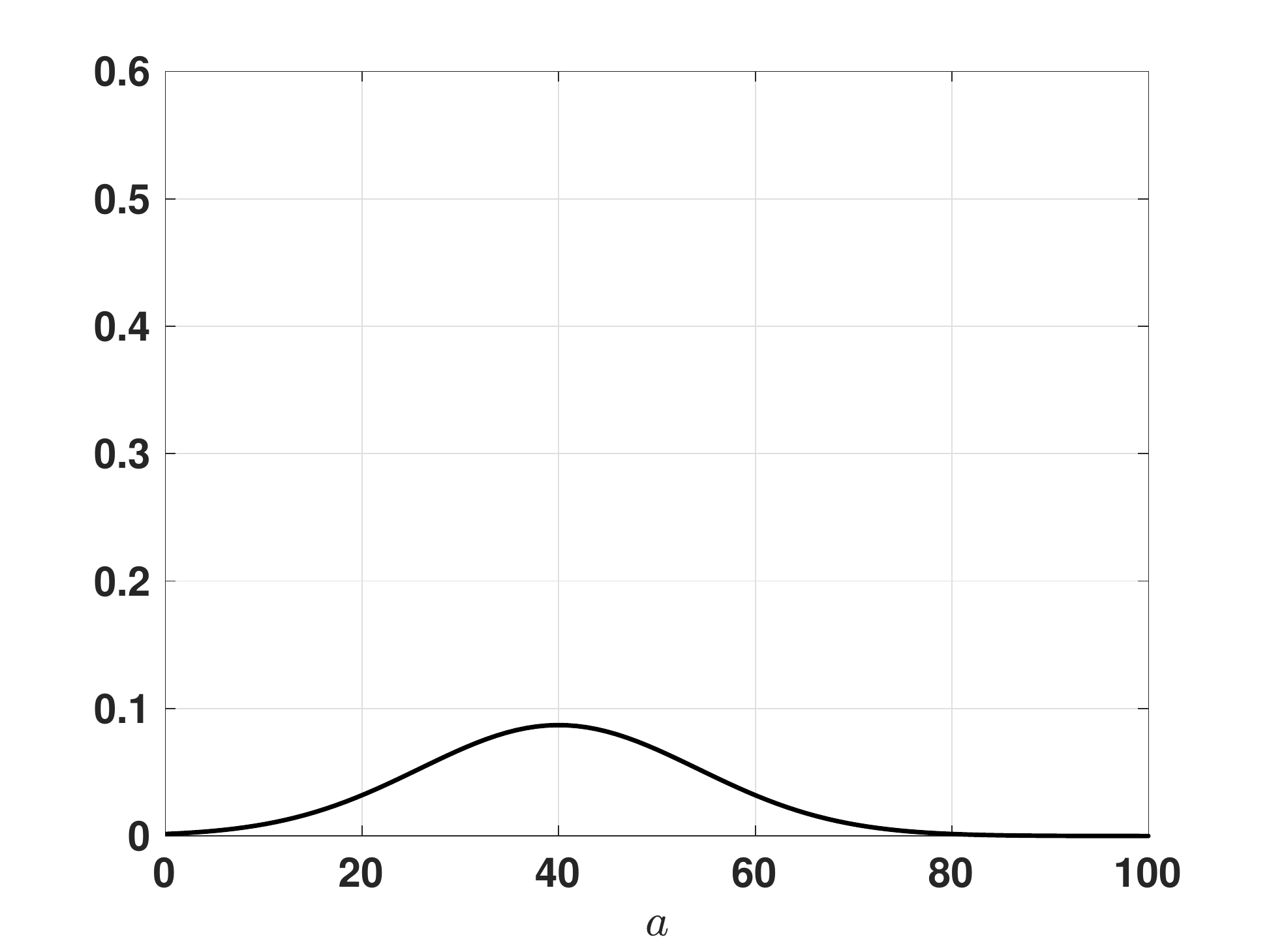}}\hfill
\subfloat[${I(10,a)}/{n(10,a)}$ for $\beta_3(a)$, $\rho_R=0$. \label{fig:ex1b3}]{\includegraphics[width=0.45\textwidth]{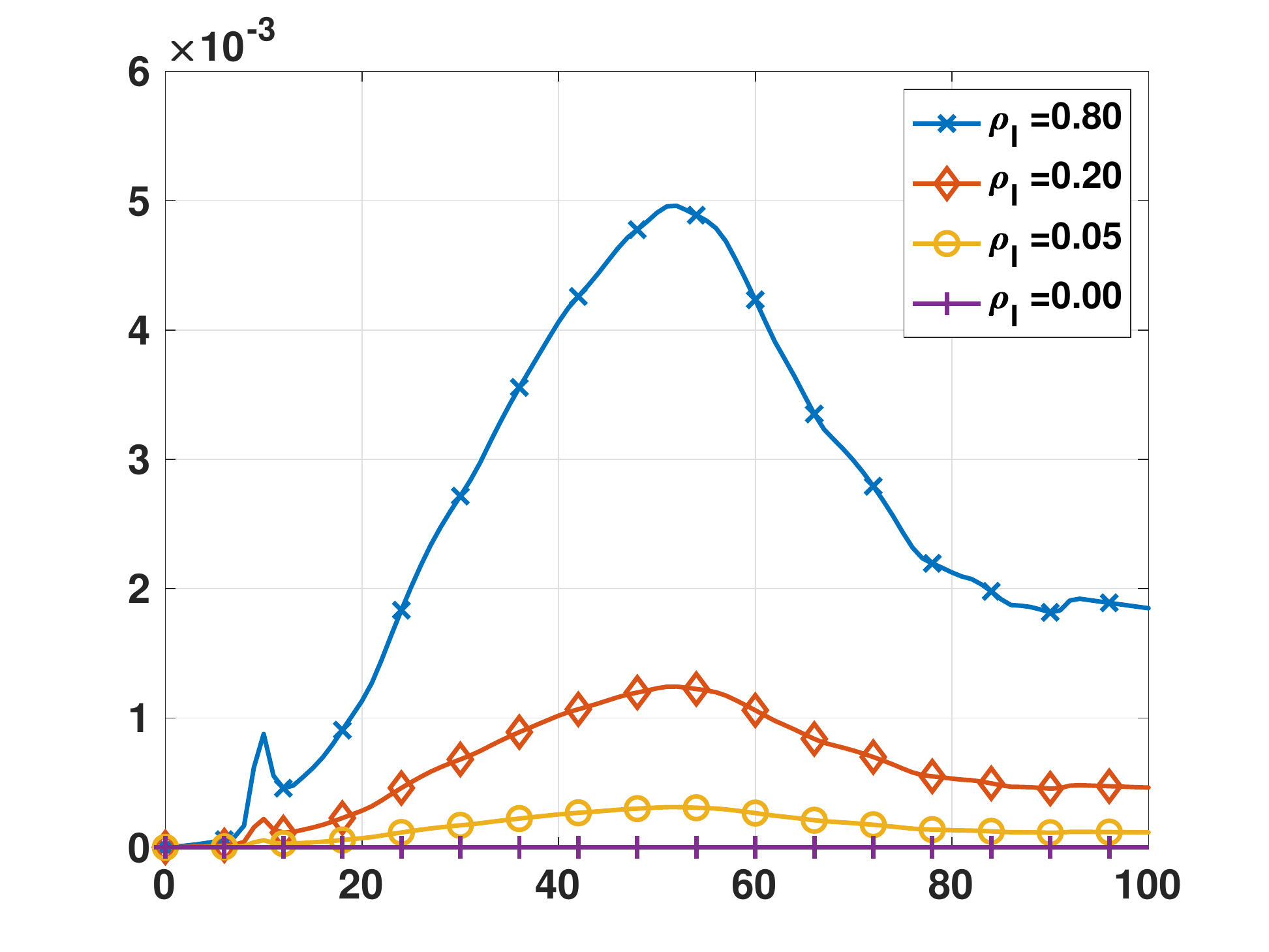}}\hfill
\caption{Three different distributions for the transmission rate, $\beta(a)$ (left), and their corresponding numerical approximation for $I(10,a)/n(10,a)$ (right); see Example \ref{ex1}.\label{fig:ex1}}
\end{figure}
We use different normal distributions for the transmission rate, $\beta(a)$, as shown in Figures \ref{fig:beta1}, \ref{fig:beta2}, \ref{fig:beta3}, for which we present the proportion of infected individuals $I(t,a)/n(t,a)$ in the population at time $t=10$, in Figures \ref{fig:ex1b1}, \ref{fig:ex1b2}, \ref{fig:ex1b3}.

In Figure \ref{fig:ex1}, we observe that for higher values of $\beta(a)$ the infected population increases.
It is clear that the infected population is zero when $\rho_I=0$. For $\rho_I=0.8$ and $\beta_1(a)$, the age distribution for the infected population has a peak close to $a=45$, and corresponds to approximately a $6.6\%$ of the total population. These are the extreme cases in this experiment. However, even if $\rho_I$ is small, for example, $\rho_I=0.05$, the number of infected individuals in the population is significant. We considered migration for a period of two years, simulating the large flow rates of immigrants due to extreme circumstances.

}
\end{example}

\begin{example} \label{ex2} 
{\rm
Consider again a susceptible population ($I_0(a)\equiv R_0(a)\equiv 0$). In this case, a proportion of immigrants come into the population as immune, that is, $\rho_R>0$. Here, we explore the possible effect of herd immunity that immigrants might provide to the population; this is, a proportion of the individuals that immigrate are immune to a particular disease. In Figure \ref{fig:ex2}, we show that there is no significant reduction in the infected population as $\rho_R$ increases.

\begin{figure}[!htb]
\centering
\subfloat[$\beta_1(a)$, $\rho_R = 0.1$.]{\includegraphics[width=0.45\textwidth]{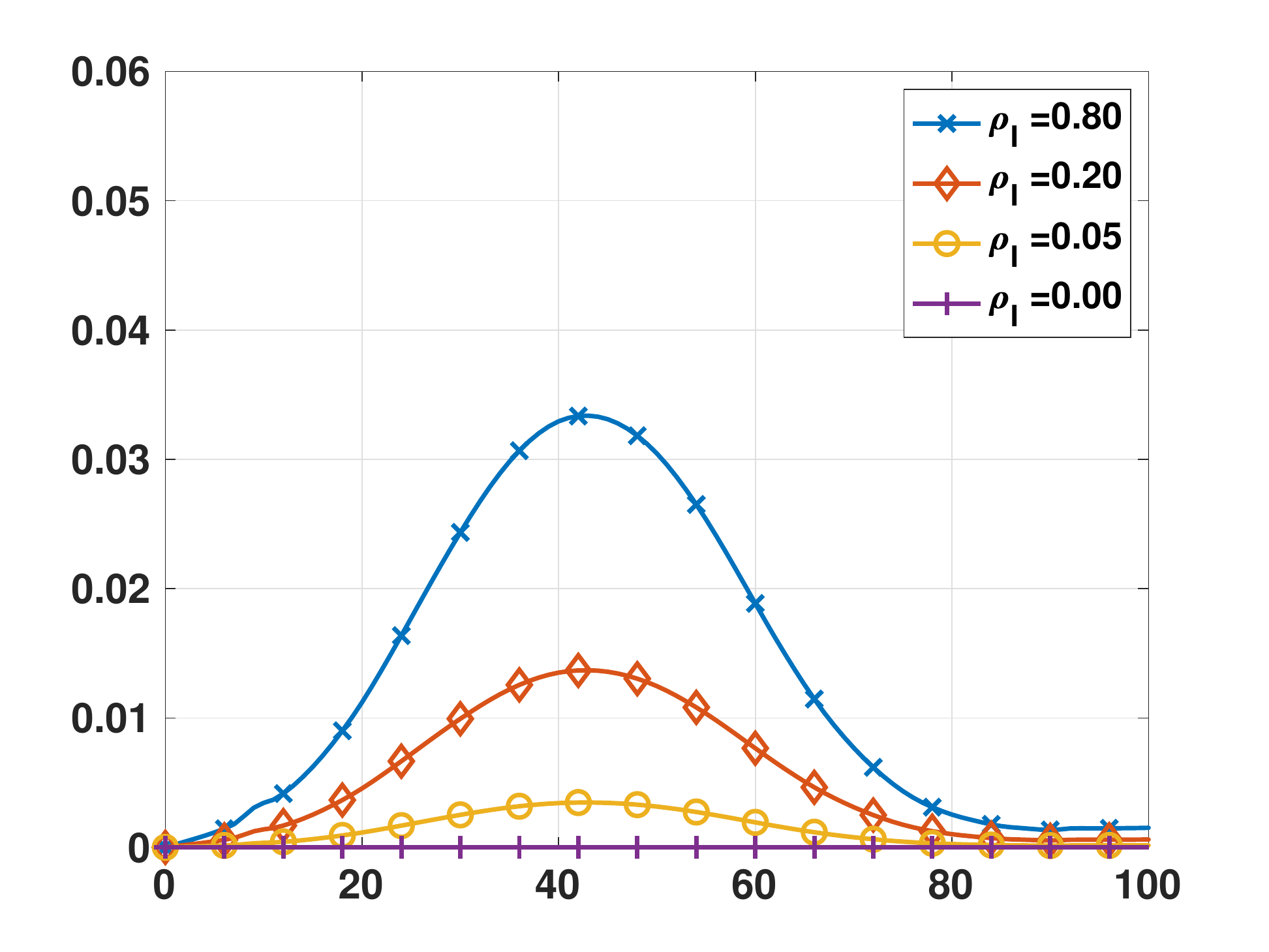}}\hfill
\subfloat[$\beta_1(a)$, $\rho_R=0.5$.]{\includegraphics[width=0.45\textwidth]{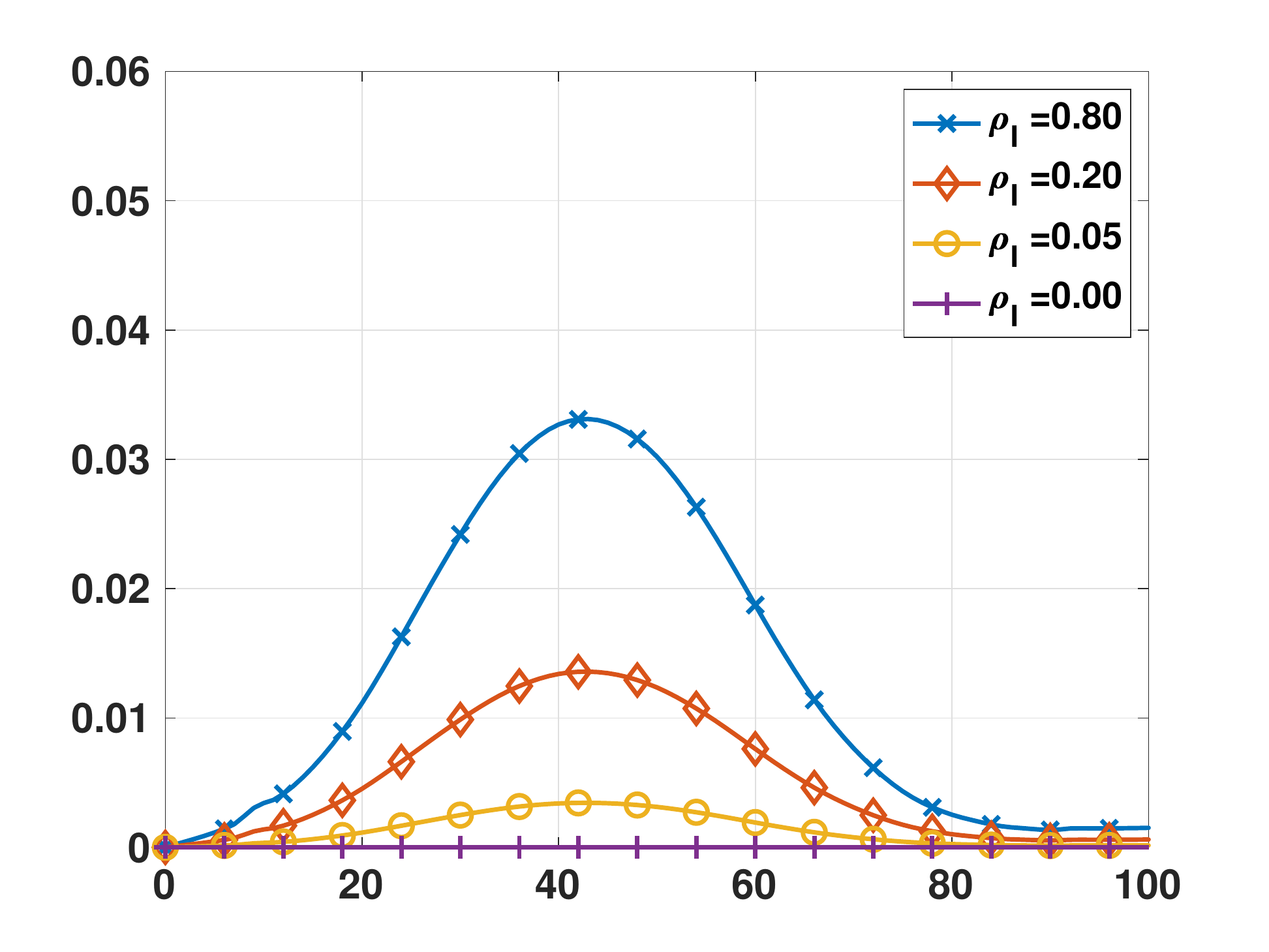}}\hfill
\subfloat[$\beta_2(a)$, $\rho_R = 0.1$.]{\includegraphics[width=0.45\textwidth]{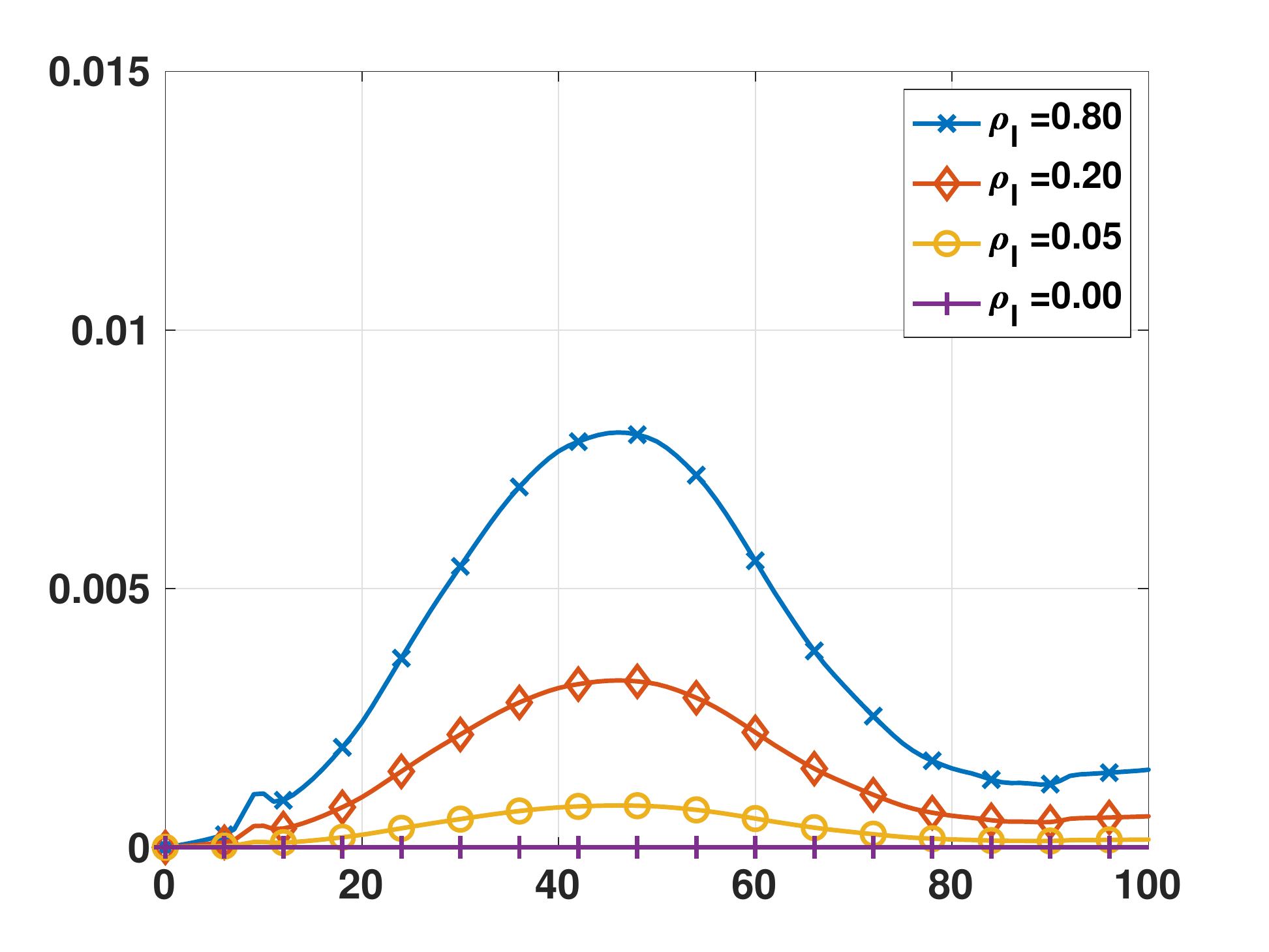}}\hfill
\subfloat[$\beta_2(a)$, $\rho_R=0.5$.]{\includegraphics[width=0.45\textwidth]{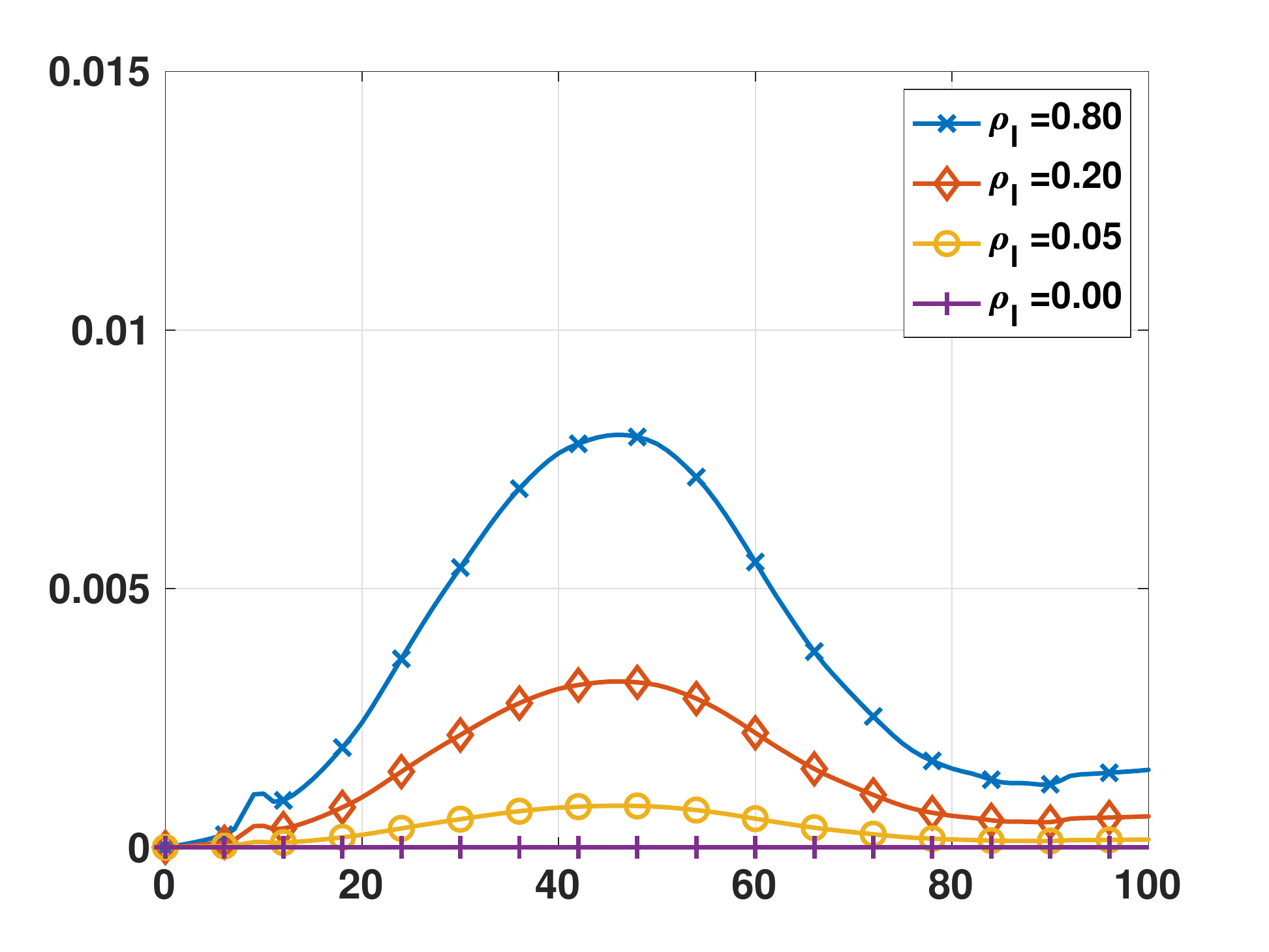}}\hfill
\subfloat[$\beta_3(a)$, $\rho_R = 0.1$.]{\includegraphics[width=0.45\textwidth]{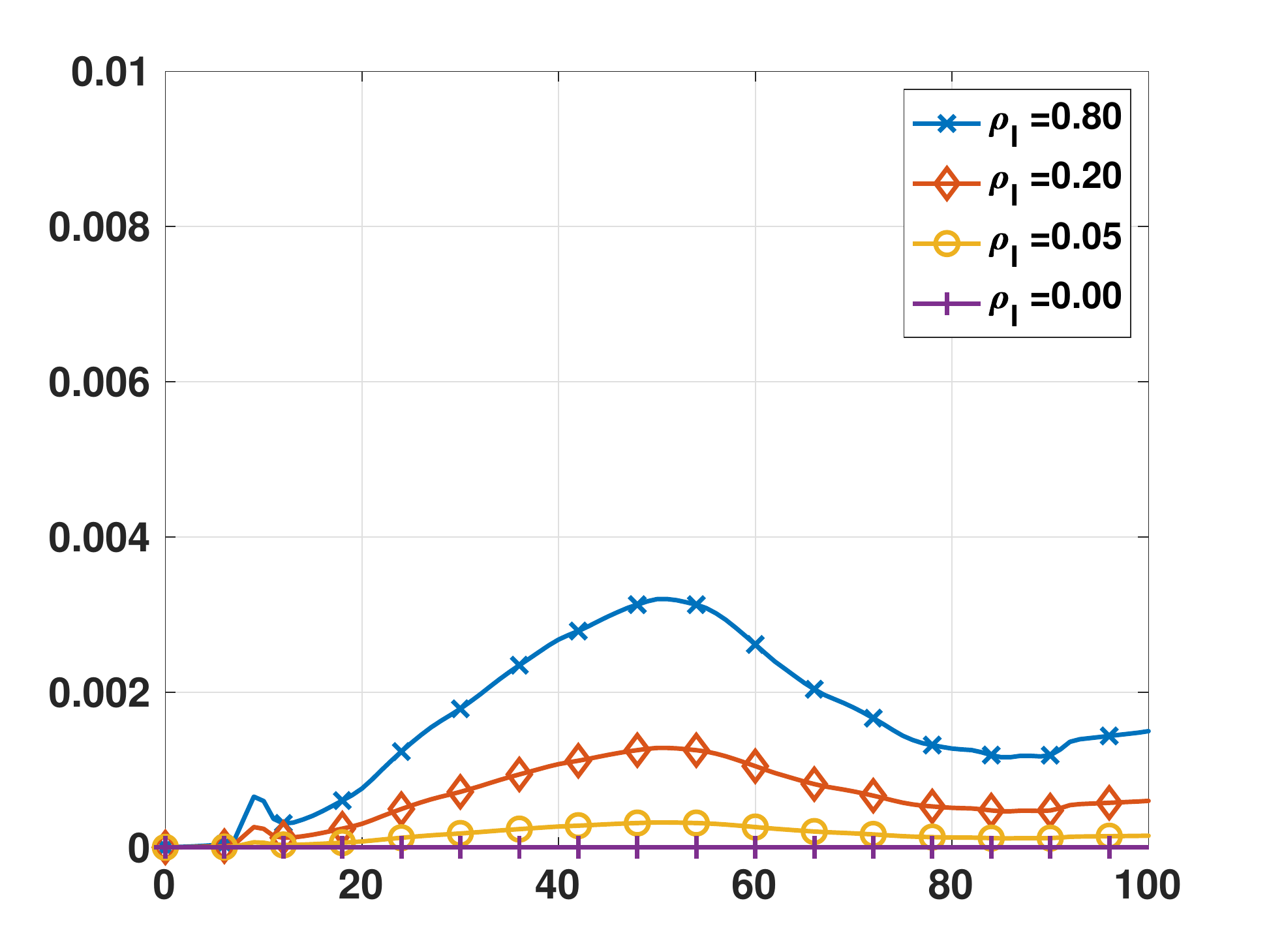}}\hfill
\subfloat[$\beta_3(a)$, $\rho_R=0.5$.]{\includegraphics[width=0.45\textwidth]{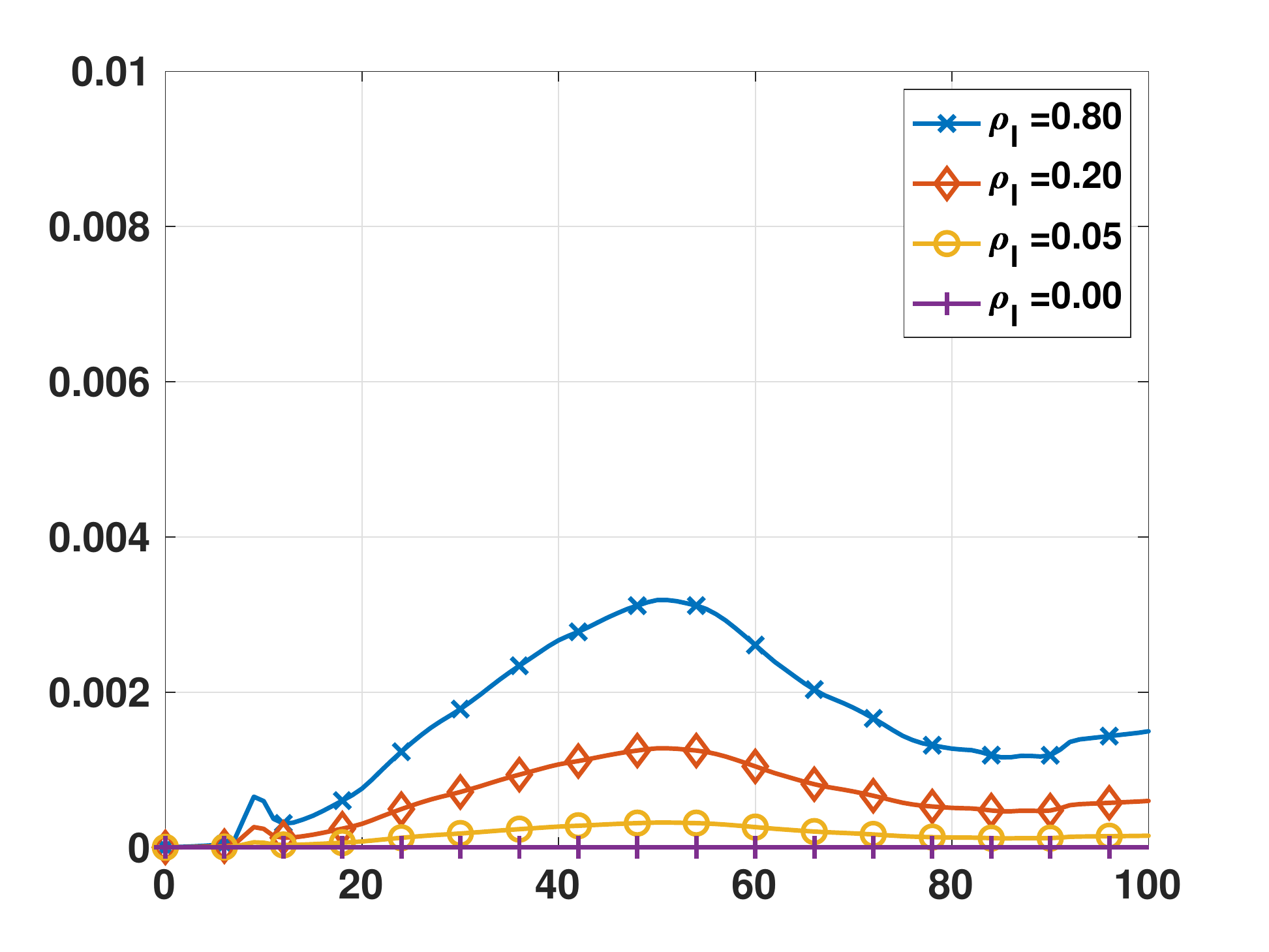}}\hfill
\caption{Values for $I(10,a)/n(10,a)$ for three different distributions of $\beta(a)$ as in Figure \ref{fig:ex1}, and $\rho_R \in \lbrace 0.1, 0.5\rbrace$; see Example \ref{ex2}.\label{fig:ex2}}
\end{figure}
}
\end{example}

\begin{example} \label{ex3}
{ \rm
Now consider an initial {\it infection-free} population with non-zero recovered class ($I_0(a)\equiv 0$, $R_0(a)\neq 0$), and infected immigrants arrive into the population. We consider $\rho_R = 0.30$, $\rho_I\in \lbrace 0, 0.10, 0.20, 0.50 \rbrace$ and different initial conditions $R_0(a) = \alpha n_0(a)$ where $\alpha \in \lbrace 0, 0.25, 0.50\rbrace$; see results in Figure \ref{fig:ex3}.

\begin{figure}[!htb]
\centering
\subfloat[${I(10,a)}/{n(10,a)}$ for $\beta_1(a)$, $\alpha=0$.\label{fig:ex3_1}]{\includegraphics[width=0.3\textwidth]{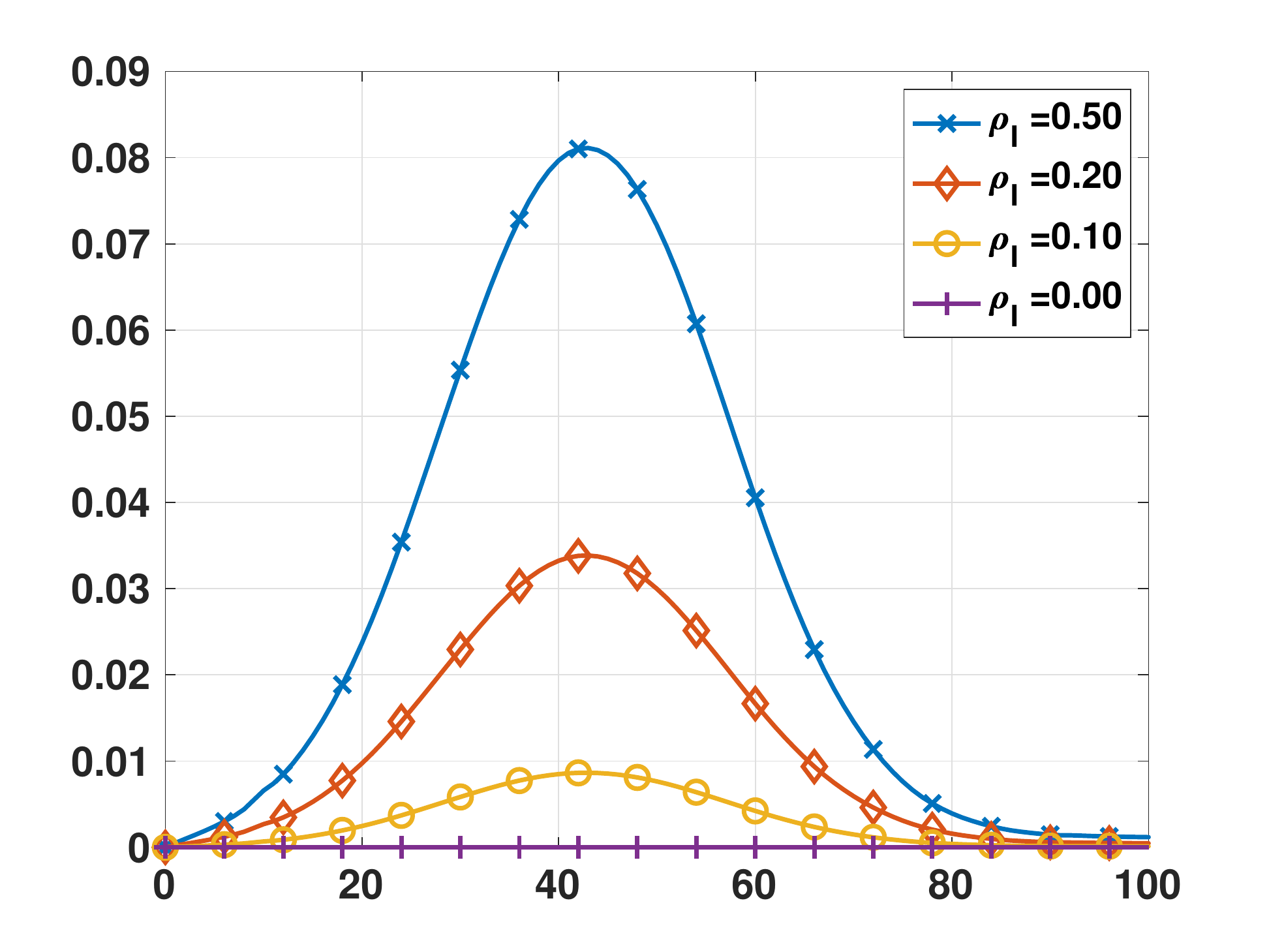}}\hfill
\subfloat[${I(10,a)}/{n(10,a)}$ for $\beta_1(a)$, $\alpha=0.25$.\label{fig:ex3_2}]{\includegraphics[width=0.3\textwidth]{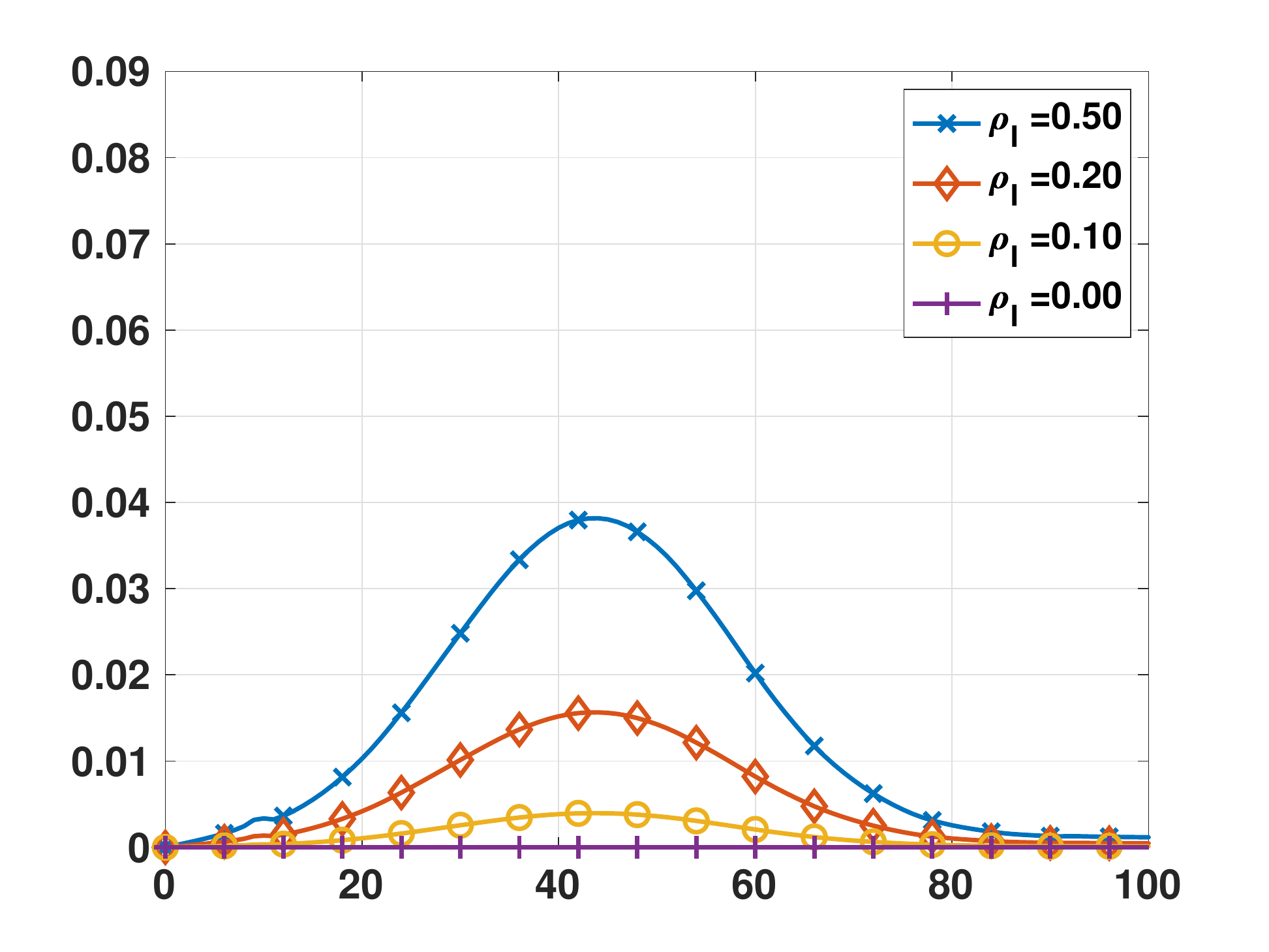}}\hfill
\subfloat[${I(10,a)}/{n(10,a)}$ for $\beta_1(a)$, $\alpha=0.50$.\label{fig:ex3_3}]{\includegraphics[width=0.3\textwidth]{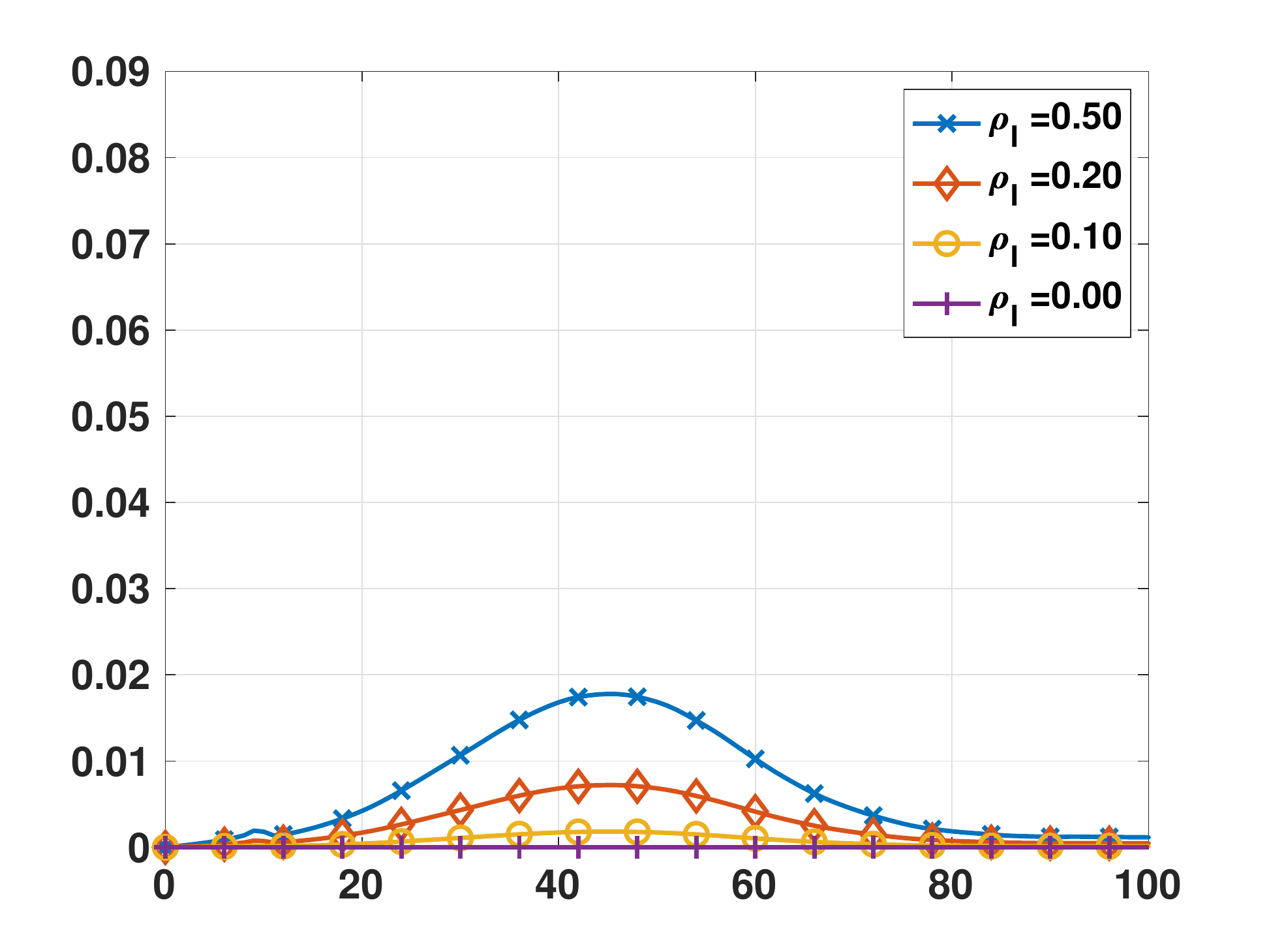}}\hfill
\caption{Values for $I(10,a)/n(10,a)$ for $\beta_1(a)$ as in Figure \ref{fig:beta1}; see Example \ref{ex3}.\label{fig:ex3}}
\end{figure}

When there is immunity of a disease in the local population, the incoming immigrant population benefits from herd immunity if a disease breaks out within the local community. In fact, both populations benefit from the effect of herd immunity. As seen in Figure \ref{fig:ex3}, the proportion of infected individuals reaches $8\%$ when the local population has no immunity. However, as the level of immunity increases the infected population has a significant decrease.
}
\end{example} 
 
\section{Conclusions} \label{sec:conc}
In this study, we focused on the overall population changes when {\it short-term immigration} is a possibility. We considered an age-structured epidemic model with {\it short-term immigration}. The {\it immigration-free} and {\it infection-free} steady state distributions 
are stable
when $\mathcal{R}_0<1$; see Section 4. The existence of the endemic non-uniform steady state distribution is guaranteed when $\mathcal{R}_0>1$ (see \cite[Theorem 4.1]{sanchez2018x}). Age-dependent parameter distributions were based on Mexico-USA immigration data \cite{mpi}. 

Numerical experiments were conducted to illustrate the distinct immigration scenarios where individuals may come as susceptible, infected or immune. Typically, $\mathcal{R}_0<1$ is a sufficient condition for a \lq\lq disease" to die out. However, when infected individuals come into a mostly susceptible population via immigration the disease becomes endemic and the incorporation of reactive strategies presents an enormous challenge to public health officials. Even for small values of $\rho_I$, the proportion of infected immigrants, the population suffers from the effect of having a new disease invade or exacerbate an existing situation for a disease which is already present. 

We explored three different scenarios where we looked at population immunity levels, both immigrant and local, and the effects it can have in the overall disease dynamics. In the first example, a proportion of the immigrant population arrives infected. No immune immigrants come into the local population. For different transmission distributions, $\beta(a)$, we looked at the effect of infected individuals coming into a susceptible population with no immunity. In Figure \ref{fig:beta1}, we see that the proportion of infected individuals reaches approximately $7\%$ when the proportion of infected immigrants coming into the local population is high ($\rho_I=0.80$). For a lower transmission rate distribution (see Figure \ref{fig:beta2}), the number of infected individuals in the population approaches $1.5\%$ when $\rho_I=0.80$. This implies that when the number of infected immigrants is high, if there are control measures that can lower the disease transmission, then the health effects that the immigrant population may have in the local presumed susceptible individuals is minimal. In contrast, in Figure \ref{fig:ex2} we have infected immigrants coming into the population, as well as immune individuals. This creates some level of herd immunity where the local population is partially protected. Even with low levels of immunity ($\rho_R=0.10$), the infected population rises slightly above $3\%$, which is the same for $\rho_R=0.50$. In Figure \ref{fig:ex3_2}, we have some level of immunity in the local population ($25\%$), as well as immunity in the immigrant population ($\rho_R=0.30$), In this example, the infected population reaches approximately $4\%$, which highlights the importance of prevention/control strategies of the incoming immigrant population to alleviate the possible harmful health effects in the local population.

When infected and immune individuals enter the new population, our simulations show that immune individuals have a modest effect on the overall dynamics of the model. This showcases the importance of prevention strategies when these episodes are imminent. 

Although the constant flow of immigrants is normal around the world, the prospect of immigrant populations is ever changing. However, the recent events in Europe, Latin America and the Caribbean have forced massive population movement in some areas. This has caused disturbances in local populations and a likely cause of concern if some of these populations may carry infectious agents that the local public health authorities are unaware of. 

Other scenarios where individuals can migrate with less strenuous conditions may impact not only disease dynamics but the complexion of a community. The case of Puerto Ricans migrating to the United States after the devastation of hurricane Maria in 2017 (where 14\% is estimated to have migrated to the mainland \cite{prcrisis}), the majority establishing their new home in Florida, is an example. The fact that Puerto Ricans are United States citizens, migration to the United States, albeit not trivial due to economic and family factors, it is a more viable alternative than other populations looking to migrate \cite{HondCaravan,Nica,Vene,Vene2,eurocrisis}. 

Albeit the record number of immigrants being detained in 2017 \cite{FU-USA}, individuals are migrating in record numbers and taking risks to seek a better life and opportunities for their families. Our model highlights the importance of the preparation of public health officials and local government agencies, that may require to look to and adapt under special circumstances such as {\it short-term immigration}. Due to limited resources, the availability of personnel to accommodate foreign individuals even for short periods is likely to cause local disruption.

The collective behaviour of individuals migrating can have lasting effects in a population, both positive and negative. It has been shown that most immigrants contribute positively to a population. However, in the short-term, if these populations carry infectious pathogens, it may be a public health threat for the local population.

\section{Acknowledgements}
The authors would like to thank the Research Center in Pure and Applied Mathematics and the Mathematics Department at Universidad de Costa Rica for their support during the preparation of this manuscript. 

\FloatBarrier

\end{document}